\documentclass[a4paper, fleqn]{cas-dc}

\usepackage[utf8]{inputenc}
\usepackage[T1]{fontenc}
\usepackage{fix-cm}

\usepackage[authoryear,sort]{natbib}

\usepackage{amsmath,amsthm,amssymb,latexsym,amscd,amsfonts,enumerate, mathtools}

\usepackage{siunitx}
\usepackage{color, fancyhdr, graphicx, epsfig}
\usepackage{appendix}
\usepackage[english]{babel}
\usepackage{multirow}
\usepackage{tabularray}
\UseTblrLibrary{booktabs}

\usepackage{soul} 
\usepackage{booktabs}
\usepackage{pifont} 

\usepackage{subcaption}
\usepackage{lineno} 
\newtheorem{theorem}{Theorem}
\newtheorem{remark}{Remark}
\newtheorem{assumption}{Assumption}

\newtheorem{definition}{Definition}
\newtheorem{problem}{Problem}

\newtheorem{lemma}{Lemma}
\usepackage{algorithm}
\usepackage{algorithmicx}
\usepackage{tikz}
\usetikzlibrary{shapes.geometric, arrows.meta, positioning, calc, fit, backgrounds}
\graphicspath{{./figs/}{../}}

\def\tsc#1{\csdef{#1}{\textsc{\lowercase{#1}}\xspace}}
\tsc{WGM}
\tsc{QE}
\tsc{EP}
\tsc{PMS}
\tsc{BEC}
\tsc{DE}

\begin{document}
	\let\WriteBookmarks\relax
	\def\floatpagepagefraction{1}
	\def\textpagefraction{.001}
	\shorttitle{Data‑Driven Inverse-RL of Linear Systems with Model Uncertainty: A Convex Optimization View}
	\let\printorcid\relax
	\shortauthors{Duc Cuong Nguyen et~al.}
	
	\title [mode = title]{Data‑Driven Inverse Reinforcement Learning of Linear Systems with Model Uncertainty: A Convex Optimization View}                      
	
	
  \author[1,2]{Duc Cuong Nguyen}
	\credit{ Writing – original draft, Software, Methodology}
    \author[1]{Phuong Nam Dao}
	\credit{ Writing – review \& editing, Supervision, Methodology}
\address[1]{School of Electrical and Electronic Engineering, Hanoi University of Science and Technology, Hanoi, Vietnam}
\address[2]{Faculty of Engineering, Lund University, Lund, Sweden}
	
	\begin{abstract}
		Inverse reinforcement learning (IRL) for linear systems seeks a cost function whose optimal controller reproduces an expert policy from data. Existing data-driven methods for discrete-time linear systems are largely built on iterative policy/value updates, repeated matrix inversions, and, in some cases, an initial stabilizing controller, which can limit numerical robustness and practical applicability. This paper develops a convex-optimization framework for data-driven inverse reinforcement learning of discrete-time linear systems with model uncertainty. For nominal systems, we derive a semidefinite characterization of inverse optimality and a relaxed formulation that recovers an equivalent state-cost matrix together with a stabilizing controller from expert trajectories. We then obtain a model-free, off-policy reformulation by replacing the unknown system matrices with a regressed kernel matrix identified from local input--state data. For uncertain local systems, we show that a standard LQR cost is generally insufficient to represent every stabilizing target gain and therefore introduce a generalized LQR cost with a state--input cross term. Based on this model, we develop a convex data-driven inverse-RL method and extend it to robust cost design over a population of perturbations via differentiable semidefinite programming and stochastic approximation. Simulations on a discrete-time power-system example show accurate recovery of expert behavior, improved robustness to gain-estimation error and model mismatch, and a simpler computational pipeline than classical iterative inverse-RL schemes.
	\end{abstract}
	
	
	

	\maketitle
	\section{Introduction}
	Inverse reinforcement learning (IRL) has become an important tool for recovering implicit performance criteria from expert behavior and then reproducing that behavior on a local system. In linear systems, this objective is closely related to inverse optimal control for quadratic costs, where the goal is to identify a cost function whose optimal controller explains a demonstrated state-feedback law. Because linear quadratic regulation (LQR) provides a precise connection among cost functions, Riccati equations, and stabilizing feedback policies, it offers a natural foundation for studying IRL in control applications \cite{lewis2012optimal}. Such an inverse viewpoint is useful for imitation control, controller interpretation, and data-driven design when expert trajectories are available but the underlying objective function is not \cite{xue2021inverse,xue2021inverseQlearning,lian2022inverse}.

	Recent work has shown that IRL can be formulated effectively for several classes of linear control problems. Inverse optimal control has been studied for tracking problems in \cite{xue2021inverse}, discrete-time inverse reinforcement Q-learning through expert imitation has been developed in \cite{xue2021inverseQlearning}, and inverse-RL formulations for linear multiplayer games have been investigated in \cite{lian2022inverse}. More recently, inverse value-iteration and Q-learning methods with stability and robustness analysis were reported in \cite{lian2024inverse}, a modified Kleinman-iteration approach was proposed in \cite{wu2026inverse}, and output-feedback discounted IRL for continuous-time linear systems was studied in \cite{wu2026output}. These works demonstrate the growing maturity of linear IRL, but they also reveal a common structural feature: many existing methods rely on repeated policy evaluation, value iteration, or policy-improvement loops.

	Although iterative IRL schemes are effective in many settings, they can be computationally demanding and numerically sensitive. Repeated matrix inversions may amplify approximation and estimation errors, and some formulations require an initial stabilizing controller, which may be unavailable in practice \cite{xue2021inverseQlearning,lian2022inverse,lian2024inverse,wu2026inverse}. In addition, when expert and local systems are not identical, the recovered target controller may fail to admit a standard inverse-LQR interpretation. These issues become particularly important in data-driven settings, where the expert gain and closed-loop matrix must themselves be estimated from finite trajectory data.

	A natural way to address these limitations is to revisit inverse RL from a convex-optimization viewpoint. For forward LQR design, the relationship among Riccati equations, linear matrix inequalities (LMIs), and convex optimization is well established \cite{boyd1994linear,balakrishnan1995connections}. This viewpoint has also influenced data-driven and model-free optimal control, including semidefinite-programming approaches to data-driven LQR \cite{rotulo2020data} and Q-function-based or primal--dual learning methods \cite{lee2018primal,farjadnasab2022model}. These developments suggest that convex structure can be exploited not only for forward controller synthesis but also for inverse RL, with the potential to reduce iteration and improve numerical robustness.

	Motivated by this observation, this paper develops a data-driven IRL framework for discrete-time linear systems with model uncertainty. We consider both nominal local systems and uncertain local systems that match the expert only at the closed-loop level. For the uncertain case, we show that a standard LQR cost may be too restrictive because not every stabilizing feedback gain can be represented as the optimal controller of a standard discrete-time LQR problem. We therefore introduce a generalized quadratic cost with a state--input cross term, following the generalized LQR framework in \cite{heij2007introduction}, and then formulate a robust cost-design problem over a population of perturbed plants. To solve the robust problem, we combine semidefinite programming with differentiable convex optimization layers \cite{agrawal2019differentiable} and stochastic approximation \cite{robbins1951stochastic}.

	The main contributions of this paper are summarized as follows.
	\begin{enumerate}
		\item We derive a model-based convex inverse-RL formulation for nominal discrete-time linear systems. The resulting semidefinite conditions recover an equivalent state-cost matrix and, through a relaxed gain-matching formulation, a stabilizing controller that reproduces the expert behavior.
		\item We develop a model-free, off-policy reformulation by replacing the unknown system matrices with a regressed kernel matrix obtained from local input--state data. Under noise-free identifiability conditions, this formulation is an exact data-driven counterpart of the nominal convex inverse-RL problem.
		\item We extend the framework to local systems with model uncertainty. In this setting, we show why a generalized LQR cost with a cross term is needed, and we derive a convex data-driven inverse-RL formulation that recovers both a generalized cost pair and a stabilizing controller.
		\item We formulate a robust inverse-RL design problem over a distribution of plant perturbations and solve it through differentiable semidefinite programming and stochastic approximation. This extension allows the learned cost to account for model variability rather than fitting a single nominal perturbation.
		\item Compared with classical two-loop and single-loop inverse-RL schemes, the proposed Algorithms~\ref{Alg 1}--\ref{Alg 3} remove repeated inverse-RL policy/value iterations and do not require an initial stabilizing gain, which leads to a simpler and more numerically stable computational pipeline.
	\end{enumerate}

	The remainder of the paper is organized as follows. Section 2 formulates the inverse-RL problems for the expert and local systems. Section 3 presents the model-based convex inverse-RL methodology. Section 4 develops the data-driven methods for nominal and uncertain systems and introduces the robust design extension. Section 5 compares the proposed framework with related data-driven inverse-RL methods, and Section 6 reports simulation results.  Finally, Section 7 concludes the paper

	\textit{Notation:} For a symmetric matrix $M$, the relations $M \succ 0$ and $M \succeq 0$ denote positive definiteness and positive semidefiniteness, respectively. The Frobenius norm is denoted by $\|\cdot\|_F$, $\mathrm{vec}(\cdot)$ stacks the columns of a matrix into a vector, and $\otimes$ denotes the Kronecker product.
	
\section{Problem Formulation and Preliminaries}\label{Section 2}

This section introduces the expert system and the nominal local system. The inverse reinforcement learning (IRL) problem is formulated for the nominal local system, while extensions to local systems with uncertainty are addressed in subsequent sections.

\subsection{Expert Target Discrete-Time System}\label{Expert_Sys}

Consider the target discrete-time (DT) linear system
\begin{equation}\label{expert_dynamic}
x_e(k+1) = A x_e(k) + B u_e(k),
\end{equation}
where $A \in \mathbb{R}^{n \times n}$, $B \in \mathbb{R}^{n \times m}$ and $x_e \in \mathbb{R}^n$, $u_e \in \mathbb{R}^m$ denote the state and control input of the expert system, respectively. Moreover, the expert input is generated by a linear state-feedback control law $u_e(k) = K_{e,1} x_e(k)$, which minimizes the infinite-horizon quadratic cost function
\begin{equation}\label{expert_cost}
V_e(x_e(k),u_e(k)) = \sum_{k=0}^{\infty} x_e^{\top}(k) Q_e x_e(k) + u_e^{\top}(k) R_e u_e(k),
\end{equation}
where $Q_e \in \mathbb{R}^{n \times n} \succeq 0$ and $R_e \in \mathbb{R}^{m \times m} \succ 0$ are the state and input weighting matrices, respectively. In addition, we assume that $(A,B)$ is controllable and that the pair $(A,\sqrt{Q_e})$ is observable. 


The minimization of \eqref{expert_cost} corresponds to a standard linear quadratic regulator (LQR) problem, according to optimal control theory in \cite{lewis2012optimal}, whose solution is given by
\begin{align}
\label{K-expert}
&K_{e,1} = -\left(R_e + B^{\top} P_e B\right)^{-1} B^{\top} P_e A \\
\label{P-expert}
&P_e = Q_e + A^{\top} P_e A - A^{\top} P_e B \left(R_e + B^{\top} P_e B\right)^{-1} B^{\top} P_e A \\
\label{V-expert}
&V_e(x_e(k)) = \min_{u_e(k)} V_e(x_e(k),u_e(k)) = x_e^{\top}(k) P_e x_e(k).
\end{align}

\subsection{Nominal Local Discrete-Time System}

Consider a nominal local system with dynamics identical to the expert system,
\begin{equation}\label{nomial_sys}
x(k+1) = A x(k) + B u(k),
\end{equation}
where $x \in \mathbb{R}^n$ and $u \in \mathbb{R}^m$ denote the state and control input of the local system, respectively. 

The local system is associated with an arbitrary infinite-horizon quadratic cost function
\begin{equation}\label{local_cost}
V(x(k),u(k)) = \sum_{k=0}^{\infty} x^{\top}(k) Q x(k) + u^{\top}(k) R u(k),
\end{equation}
where $Q \in \mathbb{R}^{n \times n} \succeq 0$ and $R \in \mathbb{R}^{m \times m} \succ 0$.

Analogously to Subsection~\ref{Expert_Sys}, the optimal control law and value function are given by
\begin{align}
\label{K-local}
&K_1 = -\left(R + B^{\top} P B\right)^{-1} B^{\top} P A, \\
\label{P-local}
&P = Q + A^{\top} P A - A^{\top} P B \left(R + B^{\top} P B\right)^{-1} B^{\top} P A, \\
\label{V-local}
&V(x(k)) = \min_{u(k)} V(x(k),u(k)) = x^{\top}(k) P x(k), \\
\label{u-local}
&u(k) = K_1 x(k).
\end{align}

\subsection{Inverse Reinforcement Learning Formulation}

The inverse reinforcement learning problem is defined under the following standard assumptions.

\begin{assumption}\label{ass1}
The system matrices $(A,B)$ in \eqref{expert_dynamic} and \eqref{nomial_sys}, the cost matrices $(Q_e,R_e)$ in \eqref{expert_cost}, and the expert controller gain $K_{e,1}$ in \eqref{K-expert} are unknown.
\end{assumption}

\begin{assumption}\label{ass2}
The trajectories $(x_e(k),u_e(k))$ of \eqref{expert_dynamic} and $(x(k),u(k))$ of \eqref{nomial_sys} are available.
\end{assumption}

\begin{definition}\label{def1}
The cost matrix $Q$ in \eqref{local_cost} is said to be equivalent to the expert cost matrix $Q_e$ in \eqref{expert_cost} if, for any given $R \in \mathbb{R}^{m \times m} \succ 0$, solving the LQR problem \eqref{K-local}–\eqref{u-local} yields a controller gain $K_1$ identical to the expert controller gain $K_{e,1}$ defined in \eqref{K-expert}.
\end{definition}

\begin{problem} \label{prop1}
Given Assumptions~\ref{ass1}–\ref{ass2}, select $R \in \mathbb{R}^{m \times m} \succ 0$ and determine a cost matrix $Q$ equivalent to $Q_e$ such that the resulting optimal controller gain coincides with the expert controller $K_{e,1}$, as defined in Definition~\ref{def1}.
\end{problem}

\section{Model-Based Inverse RL Methodology}

This section presents a model-based inverse reinforcement learning (IRL) method based on convex programming. 

The first step is to identify the target controller gain and the target closed-loop dynamics of the expert system. The target closed-loop dynamics are defined as
\begin{equation}\label{target_dyna}
F_{e,1} = A + B K_{e,1} .
\end{equation}

Both the target closed-loop matrix and the expert controller gain can be estimated from the expert trajectory data $(x_e(k), u_e(k))$ using the least-squares method described in \cite{xue2021inverse}:
\begin{align}
\label{est_target_dyna}
\hat{F}_{e,1} &= X_e^{k+1} X_e^{k\top} \left(X_e^{k} X_e^{k\top}\right)^{-1}, \\
\label{est_target_gain}
\hat{K}_{e,1} &= U_e^{k} X_e^{k\top} \left(X_e^{k} X_e^{k\top}\right)^{-1},
\end{align}
where
\begin{align}
    U_e^{k}= [u_e(k), \cdots,u_e(k+N_{\mathrm{d}}-1)], \\
    X_e^{k}= [x_e(k), \cdots,x_e(k+N_{\mathrm{d}}-1)],\\
    X_e^{k+1}= [x_e(k), \cdots,x_e(k+N_{\mathrm{d}})]]
\end{align}
Here, $N_{\mathrm{d}} > n$ denotes the number of data groups.

Suppose that the optimal cost matrix $Q$ corresponding to a selected $R \succ 0$ is obtained as a solution to Problem~\ref{prop1}. It is well known that the solution must satisfy the following conditions \cite{lian2022inverse,xue2021inverse,xue2021inverseQlearning}:
\begin{align}
\label{ric_cond}
&P = Q + A^{\top} P A - A^{\top} P B \left(R + B^{\top} P B\right)^{-1} B^{\top} P A, \\
\label{gain_cond}
&\left(R + B^{\top} P B\right) K_{e,1} =- B^{\top} P A,
\end{align}
where $K_{e,1}$ is the expert controller gain in \eqref{K-expert}.

Substituting the definition of the target closed-loop dynamics \eqref{target_dyna} into \eqref{ric_cond}–\eqref{gain_cond} yields the equivalent single condition
\begin{equation}\label{convex_cond}
P = Q + (K_{e,1})^{\top} R K_{e,1} + (F_{e,1})^{\top} P F_{e,1} .
\end{equation}

Since \eqref{convex_cond} is linear with respect to the matrix variables $P$ and $Q$, a feasibility problem can be formulated using semidefinite programming. Using the estimated quantities \eqref{est_target_dyna} and \eqref{est_target_gain} and a selected $R \succ 0$, the feasibility problem is given by

\begin{equation}\label{opt_feas}
\begin{aligned}
& \mathrm{find}
& & P,\ Q \\
& \mathrm{subject \ to}
& &
P = Q + (\hat{K}_{e,1})^{\top} R \hat{K}_{e,1} + (\hat{F}_{e,1})^{\top} P \hat{F}_{e,1}, \\
& & &
\left(R + B^{\top} P B\right)\hat{K}_{e,1} = - B^{\top} P A\\
&&&
Q \succeq 0,\quad P \succ 0, \\
& & &
\hat{K}_{e,1} \ \text{satisfies \eqref{est_target_gain}}, \quad
\hat{F}_{e,1} \ \text{satisfies \eqref{est_target_dyna}} .
\end{aligned}
\end{equation}

The feasibility problem \eqref{opt_feas} can be solved using standard convex optimization solvers such as MOSEK, SCS, or CVXOPT. Any feasible solution satisfies \eqref{convex_cond}, which is equivalent to \eqref{ric_cond} and \eqref{gain_cond}, and therefore constitutes a valid solution to Problem~\ref{prop1}.

\begin{remark} \label{remark2}
It is well known that the feasibility problem \eqref{opt_feas} admits infinitely many solutions $(P,Q)$ satisfying \eqref{ric_cond} and \eqref{gain_cond} when $\mathrm{rank}(B^{\top}) \neq n$, which is a fundamental property of inverse reinforcement learning \cite{lian2022inverse,lewis2012optimal}.
\end{remark}

Although the inverse RL solution is characterized by \eqref{convex_cond}, the feasibility of \eqref{opt_feas} depends on the accuracy of the least-squares estimates \eqref{est_target_dyna} and \eqref{est_target_gain}. Consequently, a feasible solution may not exist when these estimates are inaccurate. To address this issue, the optimization problem will be modified in to improve robustness with respect to estimation errors in the target dynamics and controller gain.
Thus, the constraint in \eqref{convex_cond} is relaxed by incorporating it into the objective function, leading to the following optimization problem:
\begin{equation}\label{opt_invRL_nominal}
\begin{aligned}
& \underset{P,Q}{\mathrm{minimize}}
& & w_1 \lVert f_1(P,Q) \rVert_F^2 + w_2 \lVert f_2(P) \rVert_F^2 \\
& \mathrm{subject\ to}
& &
Q \succeq 0,\quad P \succ 0, \\
& & &
\hat{K}_{e,1} \ \text{satisfies \eqref{est_target_gain}}, \quad
\hat{F}_{e,1} \ \text{satisfies \eqref{est_target_dyna}} .
\end{aligned}
\end{equation}
Here, $w_1$ and $w_2$ are positive scalar weights that balance the relative importance of the objective terms, $f_1(P,Q)$ and $f_2(P)$ correspond to the relaxed constraint in \eqref{convex_cond} and the optimality condition \eqref{gain_cond}, respectively, defined as
\begin{align}
    f_1(P,Q) &= P - Q + (\hat{K}_{e,1})^{\top} R \hat{K}_{e,1} + (\hat{F}_{e,1})^{\top} P \hat{F}_{e,1}\\
    f_2(P) &= \left(R + B^{\top} P B\right)\hat{K}_{e,1} + B^{\top} P A
\end{align}
The optimization problem \eqref{opt_invRL_nominal} is convex and can therefore be solved using standard convex programming techniques. If the objective of Problem~\ref{prop1} is solely to recover a cost matrix equivalent to $Q_e$, the procedure may terminate at this stage. However, to additionally obtain a stabilizing controller gain $K_{1}^{\star}$ derived from $Q$ that is close to the expert gain $K_{e,1}$, further steps are required. This setting corresponds to the inverse RL imitation problem studied in \cite{xue2021inverseQlearning}.

\begin{problem}\label{prop2}
Given Assumptions~\ref{ass1}–\ref{ass2}, select $R \in \mathbb{R}^{m \times m} \succ 0$ and determine a cost matrix $Q$ equivalent to $Q_e$ such that the resulting optimal controller gain coincides with the expert gain $K_{e,1}$, as defined in Definition~\ref{def1}. Moreover, determine a stabilizing controller gain $K_{1}^{\star}$ obtained from the optimal control problem with cost matrix $Q$ that satisfies $K_{1}^{\star} = K_{e,1}$.
\end{problem}

The optimization problem \eqref{opt_invRL_nominal} alone is insufficient to solve Problem~\ref{prop2}, as it does not explicitly characterize a stabilizing controller gain close to the expert controller $K_{e,1}$. While the estimated gain $\hat{K}_{e,1}$ is available, it does not necessarily stabilize the local system \eqref{nomial_sys}. Consequently, solving the optimal control problem associated with the recovered cost matrix $Q$ is necessary to obtain a stabilizing gain while preserving proximity to $K_{e,1}$.

The LQR problem can be interpreted as determining the tightest lower bound of the cost function \eqref{local_cost}, which can be reformulated as a linear matrix inequality (LMI), as stated in the following lemma.

\begin{lemma}\label{lem_lmi_lqr}\cite{balakrishnan1995connections}
Consider the optimization problem
\begin{equation}\label{LMI}
\begin{aligned}
& \underset{P}{\mathrm{maximize}}
& & \operatorname{tr}(P) \\
& \mathrm{subject\ to}
& &
P \succ 0, \\
& & &
\begin{bmatrix}
    A^\top P A + Q -P & A^\top P B \\
    B^\top P A & B^\top P B + R
\end{bmatrix} \succ 0 .
\end{aligned}
\end{equation}
Then, the optimal solution $P$ satisfies the following properties:
\begin{itemize}
    \item $P$ satisfies the discrete-time algebraic Riccati equation \eqref{P-local}.
    \item The controller gain $K_1 = -\left(R + B^{\top} P B\right)^{-1} B^{\top} P A$ is stabilizing and is the optimal solution to the LQR problem.
\end{itemize}
\end{lemma}

\begin{proof}
The proof of the continuous-time counterpart of this result is provided in \cite{balakrishnan1995connections}. The discrete-time case follows by analogous arguments.
\end{proof}

To enforce inverse optimality and closed-loop stability simultaneously, \eqref{opt_invRL_nominal} and \eqref{LMI} can be merged into the single convex optimization problem
\begin{equation}\label{opt_invRL_nominal_stab}
\begin{aligned}
\underset{P,Q}{\mathrm{maximize}} \quad &  \operatorname{tr}(P) \\
\mathrm{subject\ to} \quad & Q \succeq 0, \quad P \succ 0, \\
&
\begin{bmatrix}
    A^\top P A + Q- P & A^\top P B \\
    B^\top P A & B^\top P B + R
\end{bmatrix} \succ 0, \\
&\hat{K}_{e,1} \ \text{satisfies \eqref{est_target_gain}},\\
&f_2(P) = 0\\
&f_1(P,Q) = 0  
\end{aligned}
\end{equation}

\begin{theorem}\label{thm_model_based_merge}
Assume that the estimates are exact, i.e., $\hat{K}_{e,1} = K_{e,1}$ and $\hat{F}_{e,1} = F_{e,1}$. If $(P^\star,Q^\star)$ is a feasible pair of \eqref{opt_invRL_nominal_stab}, then $Q^\star$ is equivalent to $Q_e$ and
\begin{equation}
K_{1}^{\star}= -\left(R + B^{\top} P^\star B\right)^{-1} B^{\top} P^\star A
= K_{e,1}
\end{equation}
solves Problem~\ref{prop2}.
\end{theorem}

\begin{proof}
Let $(P^\star,Q^\star)$ be a feasible point of \eqref{opt_invRL_nominal_stab}. The LMI constraint in \eqref{opt_invRL_nominal_stab} is exactly the feasibility condition in Lemma~\ref{lem_lmi_lqr} with $Q=Q^\star$. Therefore, by Lemma~\ref{lem_lmi_lqr}, the gain
\begin{equation}
K_{1}^{\star}= -\left(R + B^{\top} P^\star B\right)^{-1} B^{\top} P^\star A
\end{equation}
is stabilizing and is the optimal controller for the forward LQR problem associated with $(Q^\star,R)$.

Because $f_2(P^\star)=0$ and $\hat{K}_{e,1}=K_{e,1}$, we have
\begin{equation}
\left(R + B^{\top} P^\star B\right)K_{e,1} + B^{\top} P^\star A = 0.
\end{equation}
Since $R \succ 0$ and $P^\star \succ 0$, the matrix $R + B^{\top} P^\star B$ is nonsingular. Hence,
\begin{equation}
K_{e,1}= -\left(R + B^{\top} P^\star B\right)^{-1} B^{\top} P^\star A = K_{1}^{\star}.
\end{equation}

Moreover, feasibility of \eqref{opt_invRL_nominal_stab} also gives $f_1(P^\star,Q^\star)=0$. Using the exact-estimate assumptions $\hat{K}_{e,1}=K_{e,1}$ and $\hat{F}_{e,1}=F_{e,1}$, this equality becomes
\begin{equation}
P^\star = Q^\star + K_{e,1}^{\top} R K_{e,1} + F_{e,1}^{\top} P^\star F_{e,1},
\end{equation}
which is precisely \eqref{convex_cond}. As discussed after \eqref{opt_feas}, condition \eqref{convex_cond} is equivalent to \eqref{ric_cond} and \eqref{gain_cond}. Therefore, $Q^\star$ is equivalent to $Q_e$ in the sense of Definition~\ref{def1}.
\end{proof}

Although the optimization problem \eqref{opt_invRL_nominal_stab} is theoretically sound, in practice the hard constraint $f_2(P)=0$ is often too restrictive because the estimate $\hat{K}_{e,1}$ of $K_{e,1}$ may be inaccurate. Consequently, the optimization problem may become infeasible. To address this issue, we move this constraint into the objective and minimize the residual $\lVert f_2(P) \rVert_F^2$. Since this changes the original objective, we must reformulate the problem in a way that still preserves the inverse-RL interpretation. Observe that \eqref{LMI} without its objective function is a feasibility problem whose solutions yield a stabilizing gain that can be represented within the standard LQR framework \eqref{K-expert}--\eqref{P-expert}. This leads to the following optimization problem:

\begin{equation}\label{opt_invRL_nominal_stab2}
\begin{aligned}
\underset{P}{\mathrm{minimize}} \quad &  \lVert f_2(P) \rVert_F^2 \\
\mathrm{subject\ to} \quad &   P \succ 0,\\
&
\begin{bmatrix}
    A^\top P A - P & A^\top P B \\
    B^\top P A & B^\top P B + R
\end{bmatrix} \succeq 0, \\
&\hat{K}_{e,1} \ \text{satisfies \eqref{est_target_gain}}.
\end{aligned}
\end{equation}

The corresponding candidate cost weight can then be parameterized as
\begin{equation}\label{Q_from_P}
    Q(P) = -\left(A^\top P A - P - A^\top P B \left(B^\top P B + R\right)^{-1} B^\top P A \right)
\end{equation}
To clarify how \eqref{opt_invRL_nominal_stab2} is obtained from the relaxed inverse-RL objective, we introduce the following theorem:
\begin{theorem}\label{lem_relaxed_P_reduction}
Consider the relaxed version of Problem~\ref{prop2} obtained by replacing the hard constraint $K_{1}^{\star} = K_{e,1}$ with the residual minimization $\lVert f_2(P) \rVert_F^2$. Then the optimization can be written in the single variable $P$ as \eqref{opt_invRL_nominal_stab2}. If $P^\star$ is an optimizer of \eqref{opt_invRL_nominal_stab2}, then the associated candidate cost matrix is given by
\begin{equation}\label{Q_star}
        Q^\star = -A^\top P^\star A + P^\star + A^\top P^\star B \left(B^\top P^\star B + R\right)^{-1} B^\top P^\star A .
        \end{equation}
\end{theorem}
\begin{proof}
Let 
\[S(P):=R+B^\top P B \quad K(P):=-S(P)^{-1}B^\top P A. \] 
Since $R\succ0$ and $P\succ0$, we have $S(P)\succ0$ for every feasible $P$.

For any admissible pair $(P,Q)$ in Problem~\ref{prop2}, the Riccati equation gives
\begin{equation}
P = Q + A^\top P A - A^\top P B S(P)^{-1} B^\top P A,
\end{equation}
so $Q$ is uniquely determined by $P$ as
\begin{equation}
Q(P)=P-A^\top P A + A^\top P B S(P)^{-1} B^\top P A.
\end{equation}
Thus the search over $(P,Q)$ reduces to a search over $P$, and substituting $P=P^\star$ yields \eqref{Q_star}.

Moreover,
\begin{equation}
f_2(P)=\left(R+B^\top P B\right)\hat{K}_{e,1}+B^\top P A = S(P)\bigl(\hat{K}_{e,1}-K(P)\bigr).
\end{equation}
Hence, $\lVert f_2(P)\rVert_F^2$ is exactly the relaxed gain-matching residual, while the remaining constraints are precisely the LQR feasibility conditions obtained from the Schur complement and $P \succ 0$. Therefore, the relaxed Problem~\ref{prop2} takes the form \eqref{opt_invRL_nominal_stab2}.
\end{proof}

To conclude, the general pipeline of the proposed model-based approach for solving the inverse RL problem in Problem~\ref{prop2} is summarized in Algorithm~\ref{Alg 1}.

\begin{algorithm}
    \caption{Model-Based Inverse RL Imitation Control}\label{Alg 1}
\begin{algorithmic}[1]
\State Select arbitrary matrix $R \succ 0$ 
\State Obtain $\hat{K}_{e,1}$ from \eqref{est_target_gain}.
\State Solve for $P^\star$ using the optimization \eqref{opt_invRL_nominal_stab2}.
\State Return the cost matrix $Q^\star$ computed from \eqref{Q_from_P} and the controller gain $K_{1}^{\star} = -\left(R + B^{\top} P^\star B\right)^{-1} B^{\top} P^\star A$.
\end{algorithmic}
\end{algorithm}

\section{Data-Driven Inverse RL}

Algorithm~\ref{Alg 1} provides a general pipeline for solving the inverse RL imitation control problem. However, this approach requires explicit knowledge of the system dynamics $(A,B)$, which constitutes a strong and often impractical assumption. In this section, a model-free implementation of Algorithm~\ref{Alg 1} is developed, relying solely on data trajectories from the local and expert systems. Moreover, the proposed framework is extended to accommodate local systems with model uncertainty.

\subsection{Data-Driven Inverse RL for Nominal Local Systems}\label{Data-Driven IRL}

Consider the local system driven by an arbitrary input signal:
\begin{equation}\label{offpolicy_sys}
    x(k+1) = A x(k) + B \bar{u}(k).
\end{equation}
For any matrix $P \in \mathbb{R}^{n \times n}$, the quadratic form of the next state can be expanded as
\begin{equation}\label{iden_eq}
\begin{aligned}
       &x(k+1)^\top P x(k+1) \\
       &= x(k)^\top A^\top P A x(k) + x(k)^\top A^\top P B \bar{u}(k) \\
       & + \bar{u}(k)^\top B^\top P A x(k) + \bar{u}(k)^\top B^\top P B \bar{u}(k).
\end{aligned}
\end{equation}
Define the kernel matrix $H \in \mathbb{R}^{(n+m) \times (n+m)}$ associated with $P$ as
\begin{equation}\label{kernel}
    H =
    \begin{bmatrix}
        H_{xx} & H_{xu} \\
        H_{xu}^\top & H_{uu}
    \end{bmatrix}
    =
    \begin{bmatrix}
        A^\top P A & A^\top P B \\
        B^\top P A & B^\top P B
    \end{bmatrix}.
\end{equation}
Substituting \eqref{kernel} into \eqref{iden_eq} yields
\begin{equation}\label{data_driven form}
    x(k+1)^\top P x(k+1)
    =
    \begin{bmatrix}
      x(k) \\
      \bar{u}(k)
    \end{bmatrix}^\top
    H
    \begin{bmatrix}
      x(k) \\
      \bar{u}(k)
    \end{bmatrix}.
\end{equation}
Equation \eqref{data_driven form} can be expressed directly in terms of data by introducing the regressed kernel matrix $\hat{H}$:
\begin{equation}\label{model-free-kernel}
        X^{k+1\top} P X^{k+1}
        =
        \begin{bmatrix}
      X^{k} \\
      U^{k}
    \end{bmatrix}^\top
    \hat{H}
    \begin{bmatrix}
      X^{k} \\
      U^{k}
    \end{bmatrix},
\end{equation}
where
\begin{align}
\label{local_data_u}
    U^{k} &= [\bar{u}(k), \ldots, \bar{u}(k+N_{\mathrm{d}}-1)], \\
    \label{local_data_x}
    X^{k} &= [x(k), \ldots, x(k+N_{\mathrm{d}}-1)], \\
    \label{local_data_x_plus}
    X^{k+1} &= [x(k+1), \ldots, x(k+N_{\mathrm{d}})].
\end{align}
Here, $N_{\mathrm{d}} > (n+m)(n+m+1)/2$ denotes the number of data groups.
By replacing all terms involving the unknown system dynamics in Algorithm~\ref{Alg 1} with the regressed kernel matrix $\hat{H}$, the merged optimization \eqref{opt_invRL_nominal_stab} can be reformulated in a model-free manner as


\begin{equation}\label{opt_invRL_nominal_modelfree}
\begin{aligned}
\underset{P,\hat{H}}{\mathrm{minimize}} \quad &\left\| (\hat{H}_{uu}+R)\hat{K}_{e,1} + \hat{H}_{xu}^\top \right\|_F^2 \\
\mathrm{subject\ to} \quad & P\succ 0,\\
    &\begin{bmatrix}
        \hat{H}_{xx}- P & \hat{H}_{xu} \\
        \hat{H}_{xu}^\top & \hat{H}_{uu}+R
    \end{bmatrix} \succeq 0, \\
    &P,\hat{H} \ \text{satisfy \eqref{model-free-kernel}}, \\
    &\hat{K}_{e,1} \ \text{satisfies \eqref{est_target_gain}}.
\end{aligned}
\end{equation}
For any optimizer $(P^\star,\hat{H}^\star)$ of \eqref{opt_invRL_nominal_modelfree}, the resulting controller gain and state-weighting matrix are given by
\begin{equation}\label{Q_K_from_P}
\begin{aligned}
&K_{1}^{\star} = -(\hat{H}_{uu}^\star+R)^{-1} \hat{H}_{xu}^{\star\top},\\
&Q^\star = -\hat{H}_{xx}^\star + P^\star + \hat{H}_{xu}^\star(\hat{H}_{uu}^\star+R)^{-1} \hat{H}_{xu}^{\star\top}
\end{aligned}
\end{equation}
where $P^\star$ and $\hat{H}^\star$ are obtained directly from \eqref{opt_invRL_nominal_modelfree}.

\begin{lemma}\label{prop_model_free_exact}
Assume that the data generated by \eqref{offpolicy_sys} are noise free and that
\begin{equation}
\begin{bmatrix}
X^{k} \\
U^{k}
\end{bmatrix}
\begin{bmatrix}
X^{k} \\
U^{k}
\end{bmatrix}^{\top}
\end{equation}
is nonsingular. Then every pair $(P,\hat{H})$ satisfying \eqref{model-free-kernel} also satisfies $\hat{H}=H$, where $H$ is defined by \eqref{kernel}. Consequently, \eqref{opt_invRL_nominal_modelfree} is an exact off-policy reformulation of \eqref{opt_invRL_nominal_stab}.
\end{lemma}

\begin{proof}
From \eqref{offpolicy_sys}, the left-hand side of \eqref{model-free-kernel} equals
\begin{equation}
\begin{bmatrix}
X^{k} \\
U^{k}
\end{bmatrix}^{\top}
\begin{bmatrix}
    A^\top P A & A^\top P B \\
    B^\top P A & B^\top P B
\end{bmatrix}
\begin{bmatrix}
X^{k} \\
U^{k}
\end{bmatrix}.
\end{equation}
Subtracting this expression from \eqref{model-free-kernel} gives
\begin{equation}
\begin{bmatrix}
X^{k} \\
U^{k}
\end{bmatrix}^{\top}
\left(\hat{H} -
\begin{bmatrix}
    A^\top P A & A^\top P B \\
    B^\top P A & B^\top P B
\end{bmatrix}
\right)
\begin{bmatrix}
X^{k} \\
U^{k}
\end{bmatrix}=0.
\end{equation}
The nonsingularity of the regressor Gramian implies that the only matrix satisfying the above identity is the zero matrix. Therefore $\hat{H}=H$ with $H$ defined by \eqref{kernel}, and substituting the blocks of $\hat{H}$ into \eqref{opt_invRL_nominal_modelfree} recovers \eqref{opt_invRL_nominal_stab} exactly.
\end{proof}

\begin{remark}\label{remark:persistent_excitation}
The optimization problem \eqref{opt_invRL_nominal_modelfree} is a semidefinite program and hence convex. The constraint \eqref{model-free-kernel} must yield a sufficient number of independent linear equations to uniquely identify the regressed kernel matrix $\hat{H}$. Specifically, the data matrix $[X^{k\top}, U^{k\top}] \otimes [X^{k\top}, U^{k\top}]$ must have rank at least $(n+m)(n+m+1)/2$. This condition is satisfied when the input signal $\bar{u}(k)$ persistently excites the system \eqref{offpolicy_sys}.
\end{remark}

\begin{remark}
The block LMI embedded in \eqref{opt_invRL_nominal_modelfree} can be interpreted as the primal representation of the dual approaches in \cite{farjadnasab2022model,lee2018primal}, as established in \cite{balakrishnan1995connections}.
\end{remark}

Algorithm~\ref{Alg 2} summarizes the complete model-free implementation of the proposed inverse RL framework.

\begin{algorithm}
    \caption{Model-Free Inverse RL Imitation Control}\label{Alg 2}
\begin{algorithmic}[1]
\State Select an arbitrary matrix $R \succ 0$ and estimate  $\hat{K}_{e,1}$ using  \eqref{est_target_gain}.
\State Collect data groups $(U^k, X^k, X^{k+1})$ from the local system \eqref{offpolicy_sys} driven by $\bar{u}(k)$ satisfying persistent excitation.
\State Solve for $(P,\hat{H})$ using \eqref{opt_invRL_nominal_modelfree}.
\State Return the cost matrix $Q^\star$ computed from \eqref{Q_K_from_P} and the controller gain $K_{1}^{\star} = -\left(R + \hat{H}_{uu}^\star\right)^{-1} \hat{H}_{xu}^{\star\top}$.
\end{algorithmic}
\end{algorithm}
\begin{remark}
The proposed algorithm is non-iterative and consists of a single convex optimization problem. Unlike iterative inverse RL methods \cite{xue2021inverse,xue2021inverseQlearning,lian2022inverse,wu2026output}, which involve repeated matrix inversions and are therefore susceptible to numerical instability, the proposed formulation avoids error accumulation across iterations. This structural property improves numerical robustness and reduces computational complexity. Furthermore, the method eliminates the need for an initial stabilizing gain, which is required in \cite{lian2022inverse} and may be impractical to obtain in practice.
\end{remark}

\subsection{Data-Driven Inverse RL for Local Systems with Model Uncertainty}\label{Data-Driven IRL_pertubed}

In practical scenarios, expert data are often collected from a system that shares a similar structure with the local system but is not identical. Such discrepancies may arise due to model degradation, unmodeled dynamics, or differences across production versions. As a result, it is reasonable to assume that the local system exhibits model uncertainty when compared with the expert system. Specifically, the local system dynamics are given by
\begin{equation}\label{uncer_sys}
    x(k+1) = (A + \gamma_1BD)x(k) + \gamma_2 Bu(k),
\end{equation}
where $D \in \mathbb{R}^{m \times n}$ is an unknown model perturbation matrix, and $\gamma_1, \gamma_2 > 0$ are unknown scaling factors.

In this setting, our goal is to recover a target controller gain $K_{\mathrm{tar},2}$ such that the closed-loop dynamics of the uncertain system \eqref{uncer_sys} coincide with the expert closed-loop dynamics \eqref{target_dyna}, i.e.,
\begin{equation}
A + \gamma_1 BD + \gamma_2 B K_{\mathrm{tar},2} = F_{e,1}.
\end{equation}

Under model uncertainty, the target controller recovered from data may fail to be inverse-optimal for a standard LQR cost of the form \eqref{local_cost}. To enlarge the class of admissible inverse-optimal controllers, we therefore adopt the generalized LQR cost as in \cite{heij2007introduction}
\begin{equation}\label{local_cost_cross}
\begin{aligned}
&V_N(x(k),u(k)) =\\
& \sum_{k=0}^{\infty} \left(x^{\top}(k) Q x(k) + 2 x^{\top}(k) N^{\top} u(k) + u^{\top}(k) R u(k)\right),
\end{aligned}
\end{equation}
where $N \in \mathbb{R}^{m \times n}$ and that the block matrix satisfied
\begin{equation}
\label{block_SMD}
\begin{bmatrix}
Q & N^{\top} \\
N & R
\end{bmatrix} \succeq 0
\end{equation} 
To motivate the use of the generalized LQR framework instead of the standard LQR formulation, we first establish the following lemma.

\begin{lemma}\label{prop_stabilizing_not_inverse_lqr}
Not every stabilizing feedback gain for a linear system can be represented as the optimal controller of a standard discrete-time LQR problem with $Q \succeq 0$ and $R \succ 0$.
\end{lemma}

\begin{proof}
Suppose that a feedback gain $K$ is optimal for a standard discrete-time LQR problem. Then there exist matrices $P \succ 0$, $Q \succeq 0$, and $R \succ 0$ such that
\begin{equation}\label{stationarity_condition}
(R + B^\top P B)K = B^\top P A.
\end{equation}

We now construct a stabilizing gain $K$ for which no positive definite matrix $P$ can satisfy \eqref{stationarity_condition}.

Consider the system
\[
A = \begin{bmatrix} 0.5 & 0 \\ 0 & 0.5 \end{bmatrix}, 
\quad
B = \begin{bmatrix} 1 \\ 0 \end{bmatrix},
\]
and let
\[
K = -[\,0 \;\; 0.25\,].
\]

The closed-loop matrix is
\[
A + BK = 
\begin{bmatrix} 
0.5 & 0 \\ 
0 & 0.5 
\end{bmatrix}
-
\begin{bmatrix} 
1 \\ 
0 
\end{bmatrix}
[\,0 \;\; 0.25\,]
=
\begin{bmatrix} 
0.5 & -0.25 \\ 
0 & 0.5 
\end{bmatrix},
\]
whose eigenvalues are both equal to $0.5$. Hence, $K$ is stabilizing.

Now assume, for contradiction, that $K$ is LQR-optimal. Then there exists $P \succ 0$ satisfying \eqref{stationarity_condition}. Let
\[
P = \begin{bmatrix} p_{11} & p_{12} \\ p_{12} & p_{22} \end{bmatrix} \succ 0.
\]

Compute
\[
B^\top P B = p_{11},\quad, B^\top P A = [\,0.5\,p_{11} \;\; 0.5\,p_{12}\,].
\]

Substituting into \eqref{stationarity_condition}, we obtain
\[
(R + p_{11})[\,0 \;\; 0.25\,] = [\,0.5\,p_{11} \;\; 0.5\,p_{12}\,].
\]

Equating the first components gives
\[\quad p_{11} = 0.
\]

This contradicts the requirement $P \succ 0$, which implies $p_{11} > 0$. Therefore, no such matrix $P$ exists, and $K$ cannot be an optimal LQR gain for any choice of $Q \succeq 0$ and $R \succ 0$.

Hence, not every stabilizing feedback gain admits a standard inverse-LQR representation.
\end{proof}

This lemma suggests that although the expert controller gain $K_{e,1}$ is optimal for the system \eqref{expert_dynamic} with respect to the cost function \eqref{expert_cost}, no corresponding optimal gain generally exists within the standard LQR framework for the uncertain system \eqref{uncer_sys}. This obstruction is removed by the generalized LQR formulation \eqref{local_cost_cross}. Denoting $\bar{A} := A + \gamma_1 BD$, $\bar{B}= \gamma_2 B$, the associated Riccati equation and optimal gain are
\begin{align}
P &= Q + \bar{A}^{\top} P \bar{A} - \nonumber \\
\label{ric_cond_cross}
&\quad - \left(\bar{A}^{\top} P \bar{B} + N^{\top}\right)
\left(R + \bar{B}^{\top} P \bar{B}\right)^{-1}
\left(\bar{B}^{\top} P \bar{A} + N\right), \\
\label{gain_cond_cross}
K_2 &= -\left(R + \bar{B}^{\top} P \bar{B}\right)^{-1}\left(\bar{B}^{\top} P \bar{A} + N\right).
\end{align}

The next lemma shows that every stabilizing feedback gain for the uncertain system can be realized as the optimal controller of a generalized LQR problem. In particular, there exists a target gain $K_{\mathrm{tar},2}$ such that the closed-loop system \eqref{uncer_sys} matches the target dynamics \eqref{target_dyna}, which are stable.
\begin{lemma}
For every stabilizing feedback gain $K$ of the system \eqref{uncer_sys}, there exists at least one triplet $(Q,R,N)$ satisfying the generalized LQR framework \eqref{local_cost_cross}.
\end{lemma}

\begin{proof}
Let $K$ be such that $(\bar{A}-BK)$ is stable. Then for any $W \succ 0$, there exists a unique $P \succ 0$ satisfying the Lyapunov equation:
\begin{equation}
(\bar{A}+BK)^{\top}P(\bar{A}+BK) - P = -W.
\end{equation}
To prove the lemma, we construct a triplet $(Q,N,R)$ such that $K$ satisfies the optimality condition for the generalized LQR framework. Choose an arbitrary $R \succ 0$ and define:
\begin{equation}
\begin{aligned}
S &:= R + B^{\top}PB,\\
N &:= -SK - B^{\top}P\bar{A},\\
Q &:= P - \bar{A}^{\top}P\bar{A} + K^{\top}SK.
\end{aligned}
\end{equation}
By the definition of $N$, it follows that $K = S^{-1}(B^{\top}P\bar{A} + N)$, which is the necessary condition for $K$ to be the optimal gain. Now, consider the stage cost:
\begin{equation}
\ell(x,u) = \begin{bmatrix} x \\ u \end{bmatrix}^{\top} \begin{bmatrix} Q & N^{\top} \\ N & R \end{bmatrix} \begin{bmatrix} x \\ u \end{bmatrix}.
\end{equation}
Evaluating the cost at the optimal policy $u=Kx$ yields:
\begin{equation}
\ell(x,Kx) = x^{\top}(Q - N^{\top}K - K^{\top}N + K^{\top}RK)x.
\end{equation}
Substituting the definitions of $Q$ and $N$ into the expression above, and writing $F_K := \bar{A}+BK$, we recover the Lyapunov difference:
\begin{equation}
\begin{aligned}
Q - N^{\top}K - K^{\top}N + K^{\top}RK
&= P - F_K^{\top} P F_K \\
&= W.
\end{aligned}
\end{equation}
Thus, $\ell(x,Kx) = x^{\top}Wx > 0$ for all $x \neq 0$. Since $R \succ 0$ and the cost is strictly positive at its minimum, the Schur complement condition $Q - N^{\top}R^{-1}N \succ 0$ is satisfied. This implies the block matrix
\begin{equation}
\begin{bmatrix} Q & N^{\top} \\ N & R \end{bmatrix} \succeq 0
\end{equation}
is positive semidefinite. Therefore, $(Q,R,N)$ defines a valid generalized LQR cost for which $K$ is optimal.
\end{proof}

\begin{lemma}\label{lem_cross_term}
There exists a stabilizing gain $K_{\mathrm{tar},2}$ such that the closed-loop matrix of \eqref{uncer_sys} coincides with the expert closed-loop matrix \eqref{target_dyna}.
\end{lemma}

\begin{proof}
Choose
\begin{equation}\label{cross_term_choice}
K_{\mathrm{tar},2} = \frac{1}{\gamma_2} K_{e,1} - \frac{\gamma_1}{\gamma_2} D.
\end{equation}
Then
\begin{equation}\label{cond_gain2}
A + \gamma_1 BD + \gamma_2 B K_{\mathrm{tar},2}
= A + B K_{e,1}
= F_{e,1}.
\end{equation}
Hence the closed-loop matrix of \eqref{uncer_sys} is exactly $F_{e,1}$. Since $F_{e,1}$ is stable, the gain $K_{\mathrm{tar},2}$ is stabilizing.
\end{proof}

The inverse RL imitation problem is therefore reformulated under the following assumptions and definitions.

\begin{assumption}\label{ass3}
The system matrices $(A,B)$ in \eqref{expert_dynamic} and $(A+\gamma_1,\gamma_2B)$ in \eqref{uncer_sys}, the cost matrices $(Q_e,R_e)$ in \eqref{expert_cost}, and the expert controller gain $K_{e,1}$ in \eqref{K-expert} are all unknown.
\end{assumption}

\begin{assumption}\label{ass4}
The state--input trajectory pairs $(x(k),u(k))$ of the local system \eqref{uncer_sys} are available. The expert state trajectory $x_e(k)$ is available, and the expert input trajectory $u_e(k)$ is also available.
\end{assumption}


\begin{definition}\label{def2}
The pair $(Q,N)$ in \eqref{local_cost_cross} is said to be \emph{equivalent} to the expert cost matrix $Q_e$ in \eqref{expert_cost} if, for any given $R \in \mathbb{R}^{m \times m} \succ 0$, the optimal controller gain $K_2$ obtained by minimizing \eqref{local_cost_cross} subject to the uncertain system \eqref{uncer_sys} yields a closed-loop system identical to that of the expert. That is, the resulting closed-loop matrix satisfies
\begin{equation}
\begin{aligned}\label{cond_prop2}
&A + \gamma_1 BD + \gamma_2 BK_2 = A + BK_{e,1} 
\end{aligned}
\end{equation}
\end{definition}

\begin{problem}\label{prop3}
Given Assumptions~\ref{ass3}–\ref{ass4}, select $R \in \mathbb{R}^{m \times m} \succ 0$ and define the closed-loop dynamics of the uncertain local system \eqref{uncer_sys} as
\begin{equation}
    F_2 = A + B(\gamma_1 D + \gamma_2 K_2).
\end{equation}
Determine a generalized cost pair $(Q,N)$ equivalent to $Q_e$ such that the resulting closed-loop matrix $F_2$ coincides with the expert closed-loop matrix $F_{e,1}$ in \eqref{target_dyna}, as defined in Definition~\ref{def2}. Furthermore, determine a stabilizing controller gain $K_{2}^{\star}$ obtained from the corresponding generalized optimal control problem.
\end{problem}
We show that the general pipeline for deriving a solution to Problem~\ref{prop3} closely follows that for Problem~\ref{prop2}. Specifically, we first derive the target controller gain $\hat{K}_{\mathrm{tar},2}$ and then solve a modified optimization problem to obtain the inverse-optimal weights $(Q,N)$.
However, unlike in Problem~\ref{prop2}, the target controller gain $\hat{K}_{\mathrm{tar},2}$ cannot be obtained from \eqref{est_target_gain}, because it must also satisfy the consistency condition \eqref{cond_gain2}. Deriving $\hat{K}_{\mathrm{tar},2}$ therefore normally requires knowledge of the local system dynamics in \eqref{uncer_sys}.
Instead, we avoid using the explicit system model by exploiting state--input trajectory data collected from the local system.
Consider the uncertain system \eqref{uncer_sys} driven by a persistently exciting input $\bar{u}(k)$:
\begin{equation}\label{excite_sys}
    x(k+1) = (A + \gamma_1 BD)x(k) + \gamma_2 B\bar{u}(k).
\end{equation}

Let the data matrices $U^k$, $X^k$, and $X^{k+1}$ be defined as in \eqref{local_data_u}–\eqref{local_data_x_plus}, using data collected from \eqref{excite_sys}, and suppose that the resulting data satisfy the persistent excitation condition in Remark~\ref{remark:persistent_excitation}. 
Following \cite{rotulo2020data}, the closed-loop dynamics during data collection can be expressed as
\begin{equation}
X^{k+1} = \bar{A}X^k + \bar{B}U^k.
\end{equation}

Using this representation, we obtain
\begin{equation}\label{close-loop_interm_K}
\begin{aligned}
    \bar{A} + \bar{B}\hat{K}_{\mathrm{tar},2} &= [\bar{A} \quad \bar{B}]
    \begin{bmatrix}
        I_n \\ \hat{K}_{\mathrm{tar},2}
    \end{bmatrix} \\
    &= X^{k+1}
    \begin{bmatrix}
        X^k \\ U^k
    \end{bmatrix}^{\top}
    \left(
    \begin{bmatrix}
        X^k \\ U^k
    \end{bmatrix}
    \begin{bmatrix}
        X^k \\ U^k
    \end{bmatrix}^{\top}
    \right)^{-1}
    \begin{bmatrix}
        I_n \\ \hat{K}_{\mathrm{tar},2}
    \end{bmatrix} \\
    &= [\hat{M}_1 \quad \hat{M}_2]
    \begin{bmatrix}
        I_n \\ \hat{K}_{\mathrm{tar},2}
    \end{bmatrix},
\end{aligned}
\end{equation}
where $\hat{M} = [\hat{M}_1 \quad \hat{M}_2] \in \mathbb{R}^{n \times (n+m)}$, with $\hat{M}_1 \in \mathbb{R}^{n \times n}$ and $\hat{M}_2 \in \mathbb{R}^{n \times m}$ representing the least-squares estimates of $A + \gamma_1 BD$ and $\gamma_2 B$, respectively.

Combining \eqref{cond_gain2} and \eqref{close-loop_interm_K}, the target controller gain $\hat{K}_{\mathrm{tar},2}$ can be determined via a least-squares solution:
\begin{equation}\label{final_form_target}
    \hat{K}_{\mathrm{tar},2} = (\hat{M}_2^{\top} \hat{M}_2)^{-1} \hat{M}_2^{\top} (\hat{F}_{e,1} - \hat{M}_1).
\end{equation}
\begin{remark}\label{re5}
Equation~\eqref{close-loop_interm_K} is written in a pseudo-inverse form; however, it can equivalently be interpreted as a least-squares problem when $\hat{M}_2$ is not full column rank. 
Moreover, since the optimization-based approach proposed later does not require an exact estimate of $\hat{K}_{\mathrm{tar},2}$, the method remains effective even with an inaccurate target gain.
It is also worth noting that the derivation of $\hat{K}_{\mathrm{tar},2}$ in \eqref{close-loop_interm_K}–\eqref{final_form_target} involves two pseudo-inverse (or least-squares) computations, which may be computationally expensive for large-scale systems. 
In such cases, alternative methods such as gradient descent, recursive least squares, or singular value decomposition (SVD) can be employed.
\end{remark}

The least-squares estimate $\hat{K}_{\mathrm{tar},2}$ obtained from \eqref{final_form_target} need not be stabilizing because of finite-data and regression errors. As in Subsection~\ref{Data-Driven IRL}, we therefore seek a stabilizing generalized-LQR realization whose optimal gain remains close to $\hat{K}_{\mathrm{tar},2}$. Because $\hat{M}_1$ and $\hat{M}_2$ have already been identified from data, we work directly with the estimated model and omit the data-consistency constraints in \eqref{opt_invRL_nominal_modelfree}. Define the gain-matching residual
\begin{equation}\label{gain_residual_cross}
f_3(P,N) := (\hat{M}_2^\top P \hat{M}_2 + R)\hat{K}_{\mathrm{tar},2} + (\hat{M}_2^\top P \hat{M}_1 + N).
\end{equation}
Then a convex surrogate of Problem~\ref{prop3} is given by
\begin{equation}\label{opt_invRL_nominal_stab3}
\begin{aligned}
\underset{P,N}{\mathrm{minimize}} \quad & \lVert f_3(P,N) \rVert_F^2 \\
\mathrm{subject\ to} \quad & P \succ 0, \\
&
\begin{bmatrix}
    \hat{M}_1^\top P \hat{M}_1 - P & \hat{M}_1^\top P \hat{M}_2 + N^\top \\
    \hat{M}_2^\top P \hat{M}_1 + N & \hat{M}_2^\top P \hat{M}_2 + R
\end{bmatrix} \succeq 0, \\
&\hat{K}_{\mathrm{tar},2} \ \text{satisfies \eqref{final_form_target}}.
\end{aligned}
\end{equation}

For any optimizer $(P^\star,N^\star)$ of \eqref{opt_invRL_nominal_stab3}, the corresponding inverse-optimal controller and state-weighting matrix are given by
\begin{equation}\label{Q_K_from_P_N}
\begin{aligned}
K_{2}^{\star} &= -\left(R + \hat{M}_2^{\top} P^\star \hat{M}_2\right)^{-1}
\left(\hat{M}_2^{\top} P^\star \hat{M}_1 + N^\star \right), \\
Q^\star &= P^\star - \hat{M}_1^{\top} P^\star \hat{M}_1 \\
&\quad + \left(\hat{M}_1^{\top} P^\star \hat{M}_2 + N^{\star\top}\right) \\
&\qquad \times \left(R + \hat{M}_2^{\top} P^\star \hat{M}_2\right)^{-1}
\left(\hat{M}_2^{\top} P^\star \hat{M}_1 + N^\star \right).
\end{aligned}
\end{equation}

\begin{lemma}\label{lem_relaxed_cross_reduction}
Consider the relaxed version of Problem~\ref{prop3}, where the hard gain-matching constraint is replaced by the residual minimization $\lVert f_3(P,N) \rVert_F^2$ with $f_3(P,N)$ defined in \eqref{gain_residual_cross}. Then the optimization can be written in the variables $(P,N)$ as \eqref{opt_invRL_nominal_stab3}. If $(P^\star,N^\star)$ is an optimizer of \eqref{opt_invRL_nominal_stab3}, then the corresponding candidate state-weighting matrix is given by \eqref{Q_K_from_P_N}.
\end{lemma}

\begin{proof}
The proof is identical in spirit to Lemma~\ref{lem_relaxed_P_reduction}, after replacing $(A,B)$ by the identified pair $(\hat{M}_1,\hat{M}_2)$ and retaining the cross-term variable $N$. Define
\[
\begin{aligned}
S(P) &:= R+\hat{M}_2^\top P\hat{M}_2, \\
K(P,N) &:= -S(P)^{-1}(\hat{M}_2^\top P\hat{M}_1+N),
\end{aligned}
\]
so $S(P)\succ0$ for every feasible $P$.

The generalized Riccati relation with $(\bar{A},\bar{B})$ replaced by $(\hat{M}_1,\hat{M}_2)$ uniquely determines $Q$ from $(P,N)$ as
\[
\begin{aligned}
Q(P,N) &= P-\hat{M}_1^\top P\hat{M}_1 \\
&\quad +(\hat{M}_1^\top P\hat{M}_2+N^\top)S(P)^{-1}(\hat{M}_2^\top P\hat{M}_1+N),
\end{aligned}
\]
which is exactly \eqref{Q_K_from_P_N} evaluated at $(P^\star,N^\star)$. Hence the search over $(P,Q,N)$ reduces to the variables $(P,N)$.

Moreover,
\[
\begin{aligned}
f_3(P,N) &= S(P)\hat{K}_{\mathrm{tar},2}+\hat{M}_2^\top P\hat{M}_1+N \\
&= S(P)\bigl(\hat{K}_{\mathrm{tar},2}-K(P,N)\bigr).
\end{aligned}
\]
Therefore $\lVert f_3(P,N)\rVert_F^2$ is exactly the relaxed gain-matching residual. The LMI in \eqref{opt_invRL_nominal_stab3} is the corresponding generalized-LQR feasibility condition. 
\end{proof}

\begin{algorithm}
    \caption{Model-Free Inverse RL Imitation Control for Uncertain Local Systems}\label{Alg 3}
\begin{algorithmic}[1]
\State Select an arbitrary matrix $R \succ 0$.
\State Estimate the expert closed-loop matrix $\hat{F}_{e,1}$ from expert trajectory data using \eqref{est_target_dyna}.
\State Collect data tuples $(U^k, X^k, X^{k+1})$ from the local system \eqref{excite_sys} driven by a persistently exciting input $\bar{u}(k)$.
\State Estimate $[\hat{M}_1 \; \hat{M}_2]$ from the regression in \eqref{close-loop_interm_K}.
\State Compute $\hat{K}_{\mathrm{tar},2}$ from \eqref{final_form_target} using a least-squares or pseudo-inverse implementation.
\State Solve \eqref{opt_invRL_nominal_stab3} to obtain $(P^\star,N^\star)$.
\State Recover $Q^\star$ and $K_{2}^{\star}$ from \eqref{Q_K_from_P_N}.
\State Return the generalized cost pair $(Q^\star,N^\star)$ together with the stabilizing controller gain $K_{2}^{\star}$.
\end{algorithmic}
\end{algorithm}

\subsection{Robust cost design from the nominal inverse-RL solution}

Throughout this paper, the goal of inverse reinforcement learning is not only to recover an expert controller, but also to identify a cost representation that explains the expert behavior and can be reused for related control-design tasks. This point becomes especially important for a population of systems that share the same nominal structure but are subject to different perturbations. Although the method in Subsection~\ref{Data-Driven IRL_pertubed} can be applied to each perturbed system separately, doing so is inefficient and does not exploit the common structure across the population. Motivated by this limitation, we now seek a nominal cost design that is robust in an ensemble sense, so that a single recovered cost transfers well on average across the perturbed systems.
\begin{assumption}\label{ass5}
The expert trajectory data $(x_e(k),u_e(k))$ for the system \eqref{expert_dynamic} are available and satisfy the persistent excitation condition stated in Remark~\ref{remark:persistent_excitation}.
\end{assumption}

\begin{assumption}\label{ass6}
The population of perturbed systems is described by
\begin{equation}\label{population_plant}
    x_j(k+1) = (A + B D_j)x_j(k) + Bu_j(k), \qquad j=1,\ldots,
\end{equation}
where $D_j \in \mathbb{R}^{m \times n}$ are i.i.d. random perturbation matrices with known distribution, zero mean, and finite second moment.
\end{assumption}
\begin{problem}\label{prop4}
Given Assumptions~\ref{ass5}--\ref{ass6} and a prescribed matrix $R \in \mathbb{R}^{m \times m} \succ 0$, determine a nominal generalized cost pair $(Q,N)$ such that the expected mismatch between the expert closed-loop matrix and the forward-LQR closed-loop matrices of the perturbed plants is minimized. Specifically, for each realization $D_j$, define
\begin{align}
    A_j &:= A + B D_j, \\
    K_j(Q,N) &:= -\left(R + B^{\top} P_j B\right)^{-1}\left(B^{\top} P_j A_j + N\right), \\
    \label{current_dynamic}
    F_j(Q,N) &:= A_j + B K_j(Q,N),
\end{align}
where $P_j$ is the stabilizing solution of the generalized Riccati equation
\begin{equation}
\begin{aligned}
    P_j &= Q + A_j^{\top} P_j A_j \\
    &\quad - (A_j^{\top} P_j B + N^{\top})\left(R + B^{\top} P_j B\right)^{-1}(B^{\top} P_j A_j + N).
\end{aligned}
\end{equation}
The objective is to solve
\begin{equation}\label{robust_problem_statement}
    \min_{Q,N} \; \mathbb{E}_{D_j}\!\left[\left\|F_j(Q,N)-F_{e,1}\right\|_F^2\right]
\end{equation}
subject to
\begin{equation}
\begin{bmatrix}
Q & N^{\top} \\
N & R
\end{bmatrix} \succeq 0.
\end{equation}
\end{problem}

In practice, the target matrix $F_{e,1}$ in \eqref{robust_problem_statement} can be replaced by its estimate $\hat{F}_{e,1}$ obtained from \eqref{est_target_dyna}, while the nominal matrices in the forward problem can be replaced by the data-driven estimates $(\hat{M}_1,\hat{M}_2)$ from \eqref{close-loop_interm_K}. A natural approach is therefore Monte Carlo optimization: sample perturbations $D_j$, solve the forward generalized LQR problem for each realization, and minimize the empirical version of \eqref{robust_problem_statement}. To keep the forward map numerically stable, we use the extended LMI characterization in Lemma~\ref{lem_lmi_lqr} instead of the iterative schemes used in \cite{xue2021inverse,xue2021inverseQlearning,lian2022inverse,wu2026output}, as follows:
\begin{equation}\label{LMI_extended}
\begin{aligned}
& \underset{P}{\mathrm{maximize}}
& & \operatorname{tr}(P) \\
& \mathrm{subject\ to}
& &
P \succ 0, \\
& & &
\begin{bmatrix}
    A^\top P A + Q -P & A^\top P B +N^\top\\
    B^\top P A+ N & B^\top P B + R
\end{bmatrix} \succ 0 .
\end{aligned}
\end{equation}

This optimization requires gradients of an SDP-defined forward map with respect to $(Q,N)$. Rather than differentiating through the solver iterations, we use differentiable convex optimization layers, such as CVXPYLayers \cite{agrawal2019differentiable}, and compute gradients by implicit differentiation of the optimality conditions. This approach remains numerically stable and yields exact gradients under standard regularity assumptions.

The main drawback of a full Monte Carlo formulation is its computational cost: achieving an estimation error of order $\mathcal{O}(\varepsilon)$ typically requires $\mathcal{O}(\varepsilon^{-2})$ samples. We therefore adopt a stochastic approximation approach using stochastic gradient descent (SGD), where $(Q,N)$ are updated based on gradients computed from a small batch of perturbation samples at each iteration. This substantially reduces the per-iteration cost while preserving scalability. The overall pipeline is summarized in Fig~\ref{fig:1}.

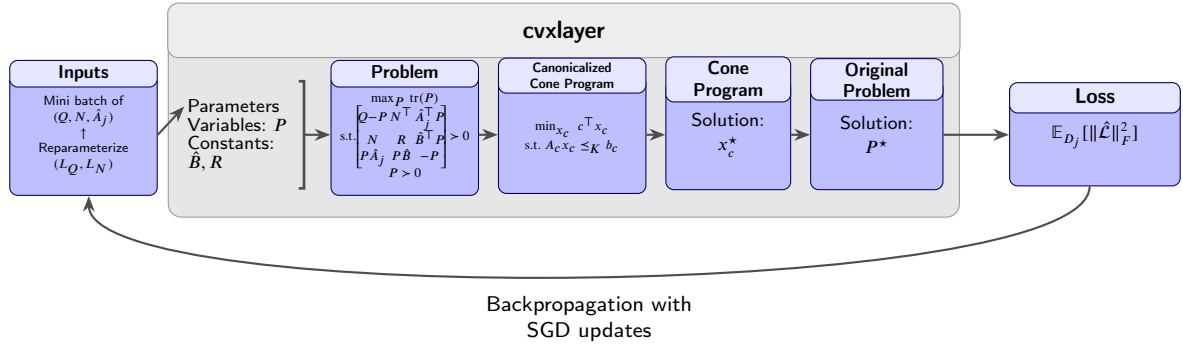
\begin{figure*}[t]
\centering
\begin{tikzpicture}[
  font=\sffamily\small,
  arrow/.style={-{Stealth[length=6pt]}, thick, draw=black!70},
  outerbox/.style={
    rectangle, rounded corners=3pt,
    draw=blue!40!black, fill=blue!25,
    minimum width=1.7cm, minimum height=1.35cm,
    align=center, inner sep=1.5pt
  },
  outertitle/.style={
    rectangle, rounded corners=3pt,
    draw=blue!40!black, fill=blue!10,
    minimum width=1.5cm, minimum height=0.4cm,
    align=center, font=\bfseries\sffamily\footnotesize,
    inner sep=0.8pt
  },
  blockbody/.style={
    rectangle, rounded corners=3pt,
    draw=blue!40!black, fill=blue!25,
    align=center, inner sep=2.5pt
  },
  blocktitle/.style={
    rectangle, rounded corners=3pt,
    draw=blue!40!black, fill=blue!10,
    align=center, font=\bfseries\sffamily\scriptsize,
    inner sep=1pt
  }
]


\node[blockbody, font=\sffamily\tiny, minimum width=1.95cm, minimum height=1.6cm] (inputs) at (0,0) {Mini batch of \\ $(Q,N,\hat{A}_j)$ \\ $\uparrow$ \\Reparameterize \\ $(L_Q,L_N)$};
\node[blocktitle, minimum width=1.95cm] (inputtitle) at ([yshift=0pt]inputs.north) {\strut Inputs};

\node[
  font=\sffamily\scriptsize,
  align=left,
  text width=1.35cm,
  inner sep=1pt,
  right=0.35cm of inputs
] (labels) {Parameters\\[-1pt]Variables: $P$\\[-1pt]Constants:\\[-1pt]$\hat{B},R$
};

\coordinate (labelbracketmid) at ([xshift=0.30cm]labels.east);
\draw[draw=black!70, line width=1pt]
  ([xshift=0.08cm,yshift=5.5pt]labels.north east) --
  ([xshift=0.15cm,yshift=5.5pt]labels.north east) --
  ([xshift=0.15cm,yshift=-5.5pt]labels.south east) --
  ([xshift=0.08cm,yshift=-5.5pt]labels.south east);

\node[blockbody, minimum width=1.95cm, minimum height=1.6cm, right=0.50cm of labels] (problem) {};
\node[blocktitle, minimum width=1.95cm] (problemtitle) at ([yshift=0pt]problem.north) {\strut Problem};
\node[font=\sffamily\tiny, align=center] at (problem.center) {%
$\max_{P}\,\mathrm{tr}(P)$\\[-1pt]
$\mathrm{s.t.}\!\left[\!\!\begin{array}{@{}c@{\,}c@{\,}c@{}}
\scriptstyle Q-P & \scriptstyle N^\top & \scriptstyle \hat{A}_j^\top P \\
\scriptstyle N & \scriptstyle R & \scriptstyle \hat{B}^\top P \\
\scriptstyle P\hat{A}_j & \scriptstyle P\hat{B} & \scriptstyle -P
\end{array}\!\!\right] \succ 0$\\[-1pt] $P\succ 0$
};

\node[blockbody, minimum width=1.95cm, minimum height=1.55cm, right=0.25cm of problem] (coneprog) {};
\node[blocktitle, minimum width=1.95cm, font=\bfseries\sffamily\tiny] (coneprogtitle) at ([yshift=0pt]coneprog.north) {\strut Canonicalized\\Cone Program};
\node[font=\sffamily\tiny, align=center] at (coneprog.center) {%
$\min_{x_c}\; c^\top x_c$\\[1pt]
$\mathrm{s.t.}\; A_c x_c \preceq_{K} b_c$
};

\node[blockbody, minimum width=1.65cm, minimum height=1.45cm, right=0.25cm of coneprog] (conesol) {};
\node[blocktitle, minimum width=1.65cm] (conesoltitle) at ([yshift=0pt]conesol.north) {\strut Cone\\Program};
\node[font=\sffamily\scriptsize, align=center] at (conesol.center) {Solution:\\[1pt] $x_c^\star$};

\node[blockbody, minimum width=1.75cm, minimum height=1.45cm, right=0.25cm of conesol] (origsol) {};
\node[blocktitle, minimum width=1.75cm] (origsoltitle) at ([yshift=0pt]origsol.north) {\strut Original\\Problem};
\node[font=\sffamily\scriptsize, align=center] (sol) at (origsol.center) {Solution:\\[1pt] $P^\star$};
\node[font=\bfseries\sffamily\small] (blankblock) at ([yshift=+9pt,] coneprogtitle.north) {};

\begin{scope}[on background layer]
  \node[
    rectangle, rounded corners=5pt,
    draw=black!40, fill=black!10,
    inner xsep=6pt,
    inner ysep=8pt,
    fit=(labels)(problemtitle)(coneprogtitle)(conesoltitle)(origsoltitle)(problem)(coneprog)(conesol)(origsol)(blankblock)
  ] (container) {};
  \path[
    draw=black!40,
    fill=black!7,
    rounded corners=4pt
  ]
    ([yshift=-0.pt]container.north west)
    rectangle
    ([yshift=-20.5pt]container.north east);
\end{scope}

\node[font=\bfseries\sffamily\small]  at ([yshift=-11pt,] container.north) {\strut cvxlayer};
\node[outerbox, anchor=west, text width=2.15cm, font=\sffamily\tiny\scriptsize] (loss) at ([xshift=0.65cm]container.east |- inputs.center) {$\mathbb{E}_{D_j}[\|\hat{\mathcal{L}}\|_F^2]$};
\node[outertitle, text width=2.2cm] (losstitle) at ([yshift = -5.5pt]loss.north) {\strut Loss};

\draw[arrow] (inputs.east) -- ([yshift=-4pt]labels.north west);
\draw[arrow] ([xshift= -4pt]labelbracketmid) -- (problem.west);
\draw[arrow] (problem.east) -- (coneprog.west);
\draw[arrow] (coneprog.east) -- (conesol.west);
\draw[arrow] (conesol.east) -- (origsol.west);
\draw[arrow] (origsol.east) -- (loss.west);

\draw[arrow]
  ([xshift=-2pt,yshift=0.0pt]loss.south)
  .. controls +(0, -1.5) and +(0, -1.5) ..
  ([xshift=2pt,yshift=0.0pt]inputs.south)
  node[midway, below=4pt, font=\sffamily\footnotesize, align=center, text width=3.8cm] {Backpropagation with\\SGD updates};

\end{tikzpicture}
\caption{General pipeline of Algorithm~\ref{Alg 4}}
\label{fig:1}
\end{figure*}

\begin{remark}\label{remark:dpp_reparam}
For a differentiable implementation based on CVXPYLayers, the constraint
$\begin{bmatrix}
Q & N^{\top} \\
N & R
\end{bmatrix} \succeq 0$
should be enforced through a factorized parameterization. Let $C$ be a fixed Cholesky factor of the prescribed matrix $R$, so that $R = CC^{\top}$, and write
\[
\begin{bmatrix}
Q & N^{\top} \\
N & R
\end{bmatrix}
=
\begin{bmatrix}
L_Q & L_N \\
0 & C
\end{bmatrix}
\begin{bmatrix}
L_Q^{\top} & 0 \\
L_N^{\top} & C^{\top}
\end{bmatrix}.
\]
Then $Q = L_Q L_Q^{\top} + L_N L_N^{\top}$ and $N = C L_N^{\top}$, so the semidefinite constraint is satisfied automatically for all factor variables $(L_Q,L_N)$. Moreover, although \eqref{LMI_extended} is convex, it is not disciplined parametrized programming (DPP) compliant. For differentiation through the SDP layer, we therefore use the equivalent DPP-compatible formulation \eqref{LMI_final}.
\begin{equation}\label{LMI_final}
\begin{aligned}
& \underset{P}{\mathrm{maximize}}
& & \operatorname{tr}(P) \\
& \mathrm{subject\ to}
& &
P \succ 0, \\
& & &
\begin{bmatrix}
    Q -P & N^\top & \hat{M}_j^\top P\\
    N &  R & \hat{M}_2^\top P \\
    P\hat{M}_j  & P\hat{M}_2 & -P
\end{bmatrix} \succ 0 .
\end{aligned}
\end{equation}
where $\hat{M}_j = \hat{M}_1+ \hat{M}_2D_j$.
\end{remark}

A practical stochastic-approximation implementation of the robust inverse-RL design is summarized in Algorithm~\ref{Alg 4}.

\begin{algorithm}
    \caption{Robust Data-Driven Inverse RL via Stochastic Approximation}\label{Alg 4}
\begin{algorithmic}[1]
\State Select a prescribed matrix $R \succ 0$, compute a Cholesky factor $C$ such that $R = CC^{\top}$, and choose a mini-batch size $N_b$, step sizes $\{\eta_t\}_{t=0}^{T-1}$, and a maximum number of SGD iterations $T$.
\State Estimate the expert closed-loop matrix $\hat{F}_{e,1}$ using \eqref{est_target_dyna} and estimate $[\hat{M}_1\;\hat{M}_2]$ using the regression in \eqref{close-loop_interm_K} from expert trajectory data.
\State Compute $\hat{K}_{\mathrm{tar},2}$ from \eqref{final_form_target} and construct initial factors $(L_Q^{(0)},L_N^{(0)})$ according to Remark~\ref{remark:dpp_reparam}.
\State \textbf{For} $t=0,1,\ldots,T-1$, repeat the following steps.
\State Form
\[
Q^{(t)} = L_Q^{(t)}L_Q^{(t)\top} + L_N^{(t)}L_N^{(t)\top},
\qquad
N^{(t)} = C L_N^{(t)\top}.
\]
\State Draw a mini-batch of perturbation samples $\{D_j\}_{j=1}^{N_b}$ from the distribution in Assumption~\ref{ass6}.
\State For each sampled perturbation, solve the forward generalized-LQR problem by the DPP-compatible SDP \eqref{LMI_final} and compute $F_j(Q^{(t)},N^{(t)})$ in \eqref{current_dynamic}.
\State Form the mini-batch empirical loss
\[
\widehat{\mathcal{L}}_t = \frac{1}{N_b}\sum_{j=1}^{N_b}\left\|F_j(Q^{(t)},N^{(t)})-\hat{F}_{e,1}\right\|_F^2.
\]
\State Compute the gradients of $\widehat{\mathcal{L}}_t$ with respect to $(L_Q,L_N)$ by implicit differentiation through the SDP layer.
\State Update $(L_Q^{(t)},L_N^{(t)})$ using an SGD step with stepsize $\eta_t$.
\State Recover $(Q^{(T)},N^{(T)})$ from the final factors and return the robust generalized cost pair $(Q^{(T)},N^{(T)})$.
\end{algorithmic}
\end{algorithm}

Algorithm~\ref{Alg 4} fits naturally into the classical stochastic-approximation framework of \cite{robbins1951stochastic}. Indeed, if we define the parameter vector by
\[
\theta_t := \begin{bmatrix}
\mathrm{vec}(L_Q^{(t)}) \\
\mathrm{vec}(L_N^{(t)})
\end{bmatrix},
\]
then the SGD step in Algorithm~\ref{Alg 4} can be written as
\[
\theta_{t+1} = \theta_t + \eta_t g(\theta_t,w_t),
\]
where $g(\theta_t,w_t) := -\nabla_{\theta} \widehat{\mathcal{L}}_t$ and $w_t$ denotes the random mini-batch of perturbation samples used at iteration $t$. Hence, the randomness in our method comes from the sampled perturbations $\{D_j\}$, while the batch loss $\widehat{\mathcal{L}}_t$ serves as a stochastic approximation of the expected objective in \eqref{robust_problem_statement}. Moreover, because $(Q,N)$ are parameterized through $(L_Q,L_N)$ with fixed $R = CC^{\top}$, the semidefinite feasibility constraint is enforced automatically, so no explicit projection term is required in our implementation. Therefore, Algorithm~\ref{Alg 4} is a projected-free SGD instance of stochastic approximation. In particular, if the step sizes satisfy
\[
\sum_{t=0}^{\infty} \eta_t = \infty, \qquad \sum_{t=0}^{\infty} \eta_t^2 < \infty,
\]
and if the associated mean dynamics $\dot{\theta} = \mathbb{E}_{w}[g(\theta,w)]$ are asymptotically stable at a fixed point $\theta^{\star}$, then classical stochastic-approximation theory implies convergence of the iterates to $\theta^{\star}$ under standard regularity assumptions.

\section{Comparison with Related Data-Driven Inverse RL Methods}

This section compares the proposed inverse-RL schemes with representative data-driven approaches in \cite{xue2021inverse,lian2022inverse,lian2024inverse,wu2026inverse}. The structural comparison is summarized in Fig.~\ref{fig:2}. 
\begin{figure}
\centering
\begin{tikzpicture}[
    scale=0.72, every node/.style={transform shape},
    block/.style={
        rectangle,
        draw,
        rounded corners=6pt,
        inner sep=3pt,
        text width=2.35cm,
        align=center,
        minimum height=1cm,
        font=\bfseries
    },
    label_style/.style={
        text width=2.2cm, 
        align=left, 
        font=\bfseries
    },
    arrow/.style={-{Stealth[scale=0.5]}, line width=0.5pt}
]

\node[label_style] (label1) {Two-loop\\inverse RL\\{\cite{xue2021inverse}}};

\node[block, right=0.0cm of label1, minimum width=1.2cm] (d1) {Data};
\node[block, right=0.4cm of d1] (r1b1) {Policy evaluation /\\ Policy update};
\node[block, right=0.5cm of r1b1] (r1b2) {Gradient descent \\ correction};
\node[block, below=0.6cm of r1b2] (r1b3) {Cost weight};
\draw[arrow] (r1b1.110) -- ++(0,0.3) -- ++(0.6,0) -- (r1b1.60);
\draw[arrow] (d1) -- (r1b1);
\draw[arrow] (r1b1) -- (r1b2);
\draw[arrow] (r1b2.east) -- ++(0.2,0) |- (r1b3.east);
\draw[arrow] (r1b3.west) -| (r1b1.south);

\node[label_style, below=1.3cm of label1] (label2) {Single-loop\\inverse RL\\{\cite{lian2022inverse,lian2024inverse,wu2026inverse}}};
\node[block, right=0.0cm of label2, minimum width=1.2cm] (d2) {Data};
\node[block, right=0.4cm of d2] (r2b1) {Policy evaluation};
\node[block, right=0.5cm of r2b1] (r2b2) {Policy update};
\node[block, below=0.6cm of r2b2] (r2b3) {Cost weight};
\draw[arrow] (d2) -- (r2b1);
\draw[arrow] (r2b1) -- (r2b2);
\draw[arrow] (r2b2.east) -- ++(0.2,0) |- (r2b3.east);
\draw[arrow] (r2b3.west) -| (r2b1.south);

\node[label_style, below=1.3cm of label2] (label3) {Iteration-free\\inverse RL\\{Alg.~\ref{Alg 1}--Alg.~\ref{Alg 3}}};
\node[block, right=0.0cm of label3, minimum width=1.2cm] (d3) {Data};
\node[block, right=0.4cm of d3] (r3b1) {One convex solve\\Value matrix / policy};
\node[block, right=0.5cm of r3b1] (r3b2) {Cost weight};
\draw[arrow] (d3) -- (r3b1);
\draw[arrow] (r3b1) -- (r3b2);

\node[label_style, below=1.3cm of label3] (label4) {Robust\\inverse RL\\{Alg.~\ref{Alg 4}}};
\node[block, right=0.0cm of label4, minimum width=1.1cm] (d4) {Data \\ Noise};
\node[block, right=0.4cm of d4] (r4b1) {Forward SDP};
\node[block, right=0.5cm of r4b1] (r4b2) {Loss / gradient};
\node[block, below=0.6cm of r4b2] (r4b3) {SGD update};
\node[block, left=0.4cm of r4b3] (r4b4) {Factors $(L_Q,L_N)$};
\draw[arrow] (d4) -- (r4b1);
\draw[arrow] (r4b1) -- (r4b2);
\draw[arrow] (r4b2.east) -- ++(0.2,0) |- (r4b3.east);
\draw[arrow] (r4b3) -- (r4b4);
\draw[arrow] (r4b4.north) -| (r4b1.south);

\end{tikzpicture}

\caption{Comparison of algorithm structures.}
\label{fig:2}
\end{figure}
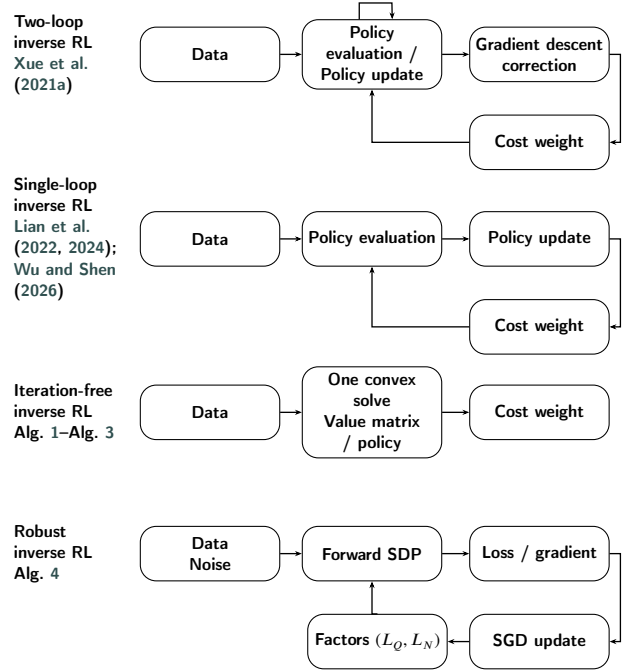

Existing data-driven inverse-RL methods can be broadly classified into two-loop and single-loop architectures. In two-loop methods, a forward RL problem is repeatedly solved inside an outer cost-update loop. Single-loop methods simplify this structure by incorporating the weight update step into the value-iteration or policy-iteration loop. In contrast, Algorithms~\ref{Alg 1}--\ref{Alg 3} proposed in this paper solve the inverse-RL problem after the required regression step through a single convex program and are therefore iteration-free at the Inverse RL stage. Algorithm~\ref{Alg 4} introduces a loop only because the robust objective is optimized by stochastic approximation.

Note that, according to \cite{wu2026inverse}, the overall computational complexity of the two-loop and single-loop algorithms are
$\mathcal{O}\big((\tfrac{n(n+1)}{2} + nm)^2 N_1\big) + \mathcal{O}\big((\tfrac{n(n+1)}{2})^2 N_2\big)$
and
$\mathcal{O}\big((n(n+1) + nm)^2 N_3\big),$
respectively, where $N_1 \geq \tfrac{n(n+1)}{2} + nm$, $N_2 \geq \tfrac{n(n+1)}{2}$, and $N_3 \geq n(n+1) + nm$.

Unlike other existing single-loop methods, the approach in \cite{wu2026inverse} merges the cost weight update and policy evaluation steps into a single stage by solving an LMI-based optimization problem. The complexity of this step is characterized as
$\mathcal{O}(\phi^{2.75} \rho^{1.5}),$
where $\phi = 1.5,n(n+1) + 2nm$ denotes the number of decision variables and $\rho = 3$ is the number of LMI constraints.

In our method, Algorithm~\ref{Alg 1} has $\phi = n(n+1)$ and $\rho = 2$. Algorithm~\ref{Alg 2} has
$\phi = \tfrac{n(n+1)}{2} + \tfrac{(n+m)(n+m+1)}{2}, \quad \rho = 2.$
However, in this case, the number of linear constraints scales with the data size $N_\mathrm{d} > \tfrac{(n+m)(n+m+1)}{2}$, which significantly affects the overall complexity. In particular, the dominant cost arises from handling these constraints, leading to a complexity of
$\mathcal{O}(\phi^2 N_\mathrm{d}).$
This makes the method computationally expensive, as each interior-point iteration requires processing $N_\mathrm{d}$ constraints, similar in structure to solving a large-scale least-squares problem. Although this optimization can become large, it is solved only once and does not require repeated policy/value iterations.

For the model-based Algorithm~\ref{Alg 3}, we have $\phi = \tfrac{n(n+1)}{2} + nm$ and $\rho = 2$. Additionally, the pseudo-inverse step used to estimate $\hat{M}$ in \eqref{close-loop_interm_K} has complexity
$\mathcal{O}((n^2 + mn)^2 N_\mathrm{d}),$
when computed via SVD \cite{golub2013matrix}. Since this step is decoupled from the optimization, it can be accelerated using more efficient techniques (see Remark~\ref{re5}), making it more scalable with respect to data compared to the model-free approach.

Finally, Algorithm~\ref{Alg 4} has $\phi = \tfrac{n(n+1)}{2} + nm$ and $\rho = 2$ per forward pass. Moreover, as discussed in \cite{agrawal2019differentiable}, the backward pass obtained by implicit differentiation requires solving a linear system derived from the optimality conditions, and its cost is of the same order as that of the forward pass.

\begin{remark}
Note that in Algorithm~\ref{Alg 4}, the matrix inversion of $R + B^{\top} P_j B$ in the discrete-time gain expression has complexity $\mathcal{O}(m^3)$ for each forward solve. This cost could be reduced in a continuous-time counterpart of the method, where the optimal gain typically takes the form $K_j = R^{-1}(B^{\top} P_j + N)$. In that case, the only matrix that must be inverted is $R$, which is fixed and can therefore be computed once in advance and reused throughout the optimization loop.
\end{remark}

Overall, the main advantage of the proposed framework is structural simplicity. Algorithms~\ref{Alg 1}--\ref{Alg 3} remove the repeated inverse-RL policy/value updates that appear in existing two-loop and single-loop methods, while Algorithm~\ref{Alg 4} introduces iteration only for robust optimization. This preserves the numerical stability of the convex-optimization formulation and yields a simpler computational pipeline than classical iterative inverse-RL schemes.

\section{Simulation}
This section validates the proposed methods on the discrete-time power-system example in \cite{vamvoudakis2015asymptotically}:
\begin{equation*}
\begin{bmatrix}
\Delta\dot{\zeta} \\
\Delta\dot{\mathcal{P}}_m \\
\Delta\dot{f}_G
\end{bmatrix} = 
\begin{bmatrix}
0.882 & 0.001 & 0.046\\
0.111 & 0.904 & 0.002 \\
0.003 & 0.057 & 0.999
\end{bmatrix}
\begin{bmatrix}
\Delta\zeta \\
\Delta \mathcal{P}_m \\
\Delta f_G
\end{bmatrix} + 
\begin{bmatrix}
 0.117\\
 0.007\\
 0.000
\end{bmatrix} \Delta P_c,
\end{equation*}
where $\Delta\zeta$, $\Delta \mathcal{P}_m$, $\Delta f_G$, and $\Delta P_c$ denote the deviations in governor valve position, generator mechanical power output, frequency, and control input, respectively. The expert weighting matrices are chosen as $Q_e = I_3$ and $R_e = 1$. The corresponding expert gain is
\[
K_{e,1} = -[0.869 \quad 0.807 \quad 1.395].
\]
All convex optimization problems are solved in the CVXPY framework with the MOSEK solver, except Algorithm~\ref{Alg 4}, for which the ECOS solver is used.
\subsection{Nominal cases}
In the nominal case, we compare the methods proposed in this paper with the discrete-time inverse-RL algorithm in \cite{xue2021inverseQlearning}, using a stopping tolerance of $10^{-8}$. The formulations \eqref{opt_feas} and \eqref{opt_invRL_nominal} solve Problem~\ref{prop1}, whereas \eqref{opt_invRL_nominal_stab} and \eqref{opt_invRL_nominal_modelfree} solve Problem~\ref{prop2}. Table~\ref{tab:nominal-performance} reports the runtime and the gain-recovery error $\|\Delta K^\star\|_F$ for different levels of estimation error in $\hat{K}_{e,1}$.

\begin{table*}[t]
\centering
\caption{Performance comparison in the nominal case.}
\label{tab:nominal-performance}
\small
\setlength{\tabcolsep}{5pt}
\renewcommand{\arraystretch}{1.1}
\begin{tabular}{llccccc}
\toprule
\multirow{2}{*}{Case} & \multirow{2}{*}{Metric} & \multicolumn{5}{c}{Algorithms} \\
\cmidrule(lr){3-7}
& & \eqref{opt_feas} & \eqref{opt_invRL_nominal} & \eqref{opt_invRL_nominal_stab} & \eqref{opt_invRL_nominal_modelfree} & \cite{xue2021inverseQlearning} \\
\midrule
\multirow{2}{*}{\makecell{$\|\hat{K}_{e,1}-K_{e,1}\|_F$ \\ $=0$}} & Run time & $0.13$ s & $0.20$ s & $0.18$ s & $0.50$ s & $0.35$ s \\
& $\|\Delta K^\star\|_F$ & $10^{-15}$ & $10^{-15}$ & $10^{-15}$ & $10^{-6}$ & $0.74$ \\
\midrule
\multirow{2}{*}{\makecell{$\|\hat{K}_{e,1}-K_{e,1}\|_F$ \\ $=0.1$}} & Run time & $0.40$ s & $0.24$ s & $0.25$ s & $0.40$ s & $0.44$ s \\
& $\|\Delta K^\star\|_F$ & $0.15$ & $0.15$ & $0.15$ & $10^{-4}$ & $0.78$ \\
\midrule
\multirow{2}{*}{\makecell{$\|\hat{K}_{e,1}-K_{e,1}\|_F$ \\ $=0.5$}} & Run time & $0.20$ s & $0.30$ s & $0.12$ s & $0.60$ s & $0.50$ s \\
& $\|\Delta K^\star\|_F$ & $0.63$ & $0.63$ & $0.63$ & $1.7\times 10^{-5}$ & $0.82$ \\
\midrule
\multirow{2}{*}{\makecell{$\|\hat{K}_{e,1}-K_{e,1}\|_F$ \\ $=2.5$}} & Run time & $0.30$ s & $0.15$ s & $0.20$ s & $0.40$ s & $0.50$ s \\
& $\|\Delta K^\star\|_F$ & $\mathrm{inf}$ & $3.02$ & $\mathrm{inf}$ & $1.57$ & $3.62$ \\
\bottomrule
\end{tabular}
\end{table*}

Table~\ref{tab:nominal-performance} shows that when the estimate $\hat{K}_{e,1}$ is exact, the proposed optimization-based methods recover the expert gain to numerical precision, whereas the benchmark in \cite{xue2021inverseQlearning} still exhibits a relatively large gain error. As the estimation error increases to $0.1$ and $0.5$, the relaxed model-free formulation \eqref{opt_invRL_nominal_modelfree} remains highly accurate, while the exact formulations become increasingly sensitive to estimation error. In the large-error case $\|\hat{K}_{e,1}-K_{e,1}\|_F=2.5$, \eqref{opt_feas} and \eqref{opt_invRL_nominal_stab} become infeasible, whereas the relaxed formulations remain solvable and achieve smaller gain error than the benchmark. These results highlight the benefit of relaxing the gain-matching constraint when the expert-gain estimate is inaccurate.

\subsection{Perturbed case}

We next consider the perturbed system \eqref{uncer_sys} with $D = [1.105 \quad -1.702 \quad -2.888]$ and $\gamma_1 = \gamma_2 = 2$. From Lemma~\ref{lem_cross_term}, the estimated expert gain for this system is
\[
\hat{K}_{e,2}= [1.539 \quad -1.298 \quad -2.191].
\]
Applying Algorithm~\ref{Alg 3} with $R=1$ yields the following cost weights:
\[
Q=\begin{bmatrix}
3.149 & -0.141&  0.267 \\
-0.141 & 1.906& 0.477 \\
 0.267&   0.477&  0.584
\end{bmatrix}
  \]
\[
N= [0.881\quad -1.859 \quad -2.672].
\]
The resulting controller achieves $\|K^\star_{e,2}- \hat{K}_{e,2} \|= 1.6 \times 10^{-6}$ and $\|\bar{A}+\bar{B}K^\star- F_{e,2} \|= 3.8 \times 10^{-7}$. In contrast, the method in \cite{xue2021inverseQlearning} yields an inverse controller with $\|K^\star_{e,2}- \hat{K}_{e,2} \|= 10.788$ and $\|\bar{A}+\bar{B}K^\star- F_{e,2} \|= 3.409$. This large performance gap is expected because, in the perturbed case, the target gain need not correspond exactly to an LQR-optimal controller; it may only be a stabilizing gain for \eqref{uncer_sys}. Figure~\ref{fig:3} compares the expert trajectories with those of the local system driven by the proposed inverse-RL controller from the same initial condition. The local responses closely follow the expert trajectories for all three states, which confirms that the recovered controller preserves the desired closed-loop behavior despite the perturbation.

\begin{figure*}
    \centering
    \begin{subfigure}[b]{0.33\textwidth}
        \centering
        \includegraphics[width=\linewidth]{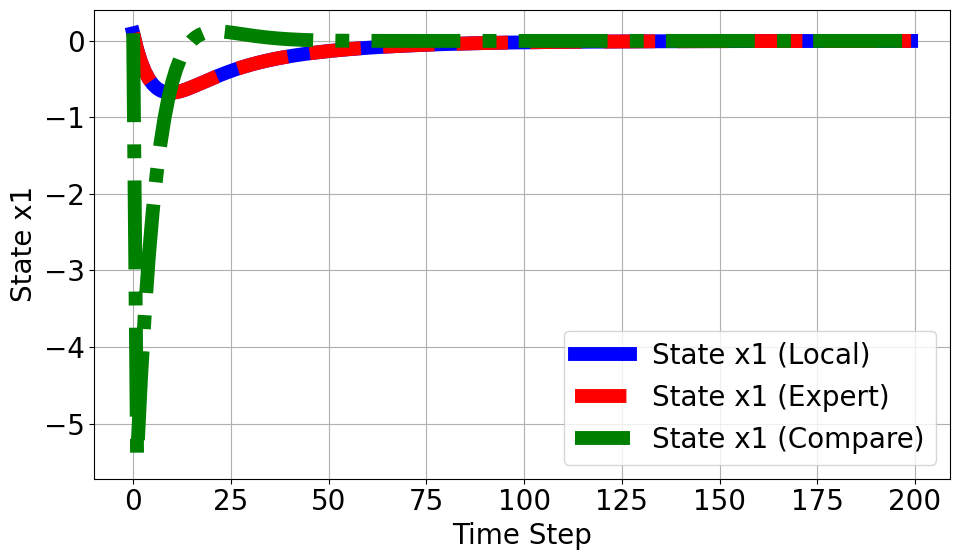}
        \caption{Governor valve position $\Delta \zeta$}
        \label{fig:3a}
    \end{subfigure}%
    \hfill
    \begin{subfigure}[b]{0.33\textwidth}
        \centering
        \includegraphics[width=\linewidth]{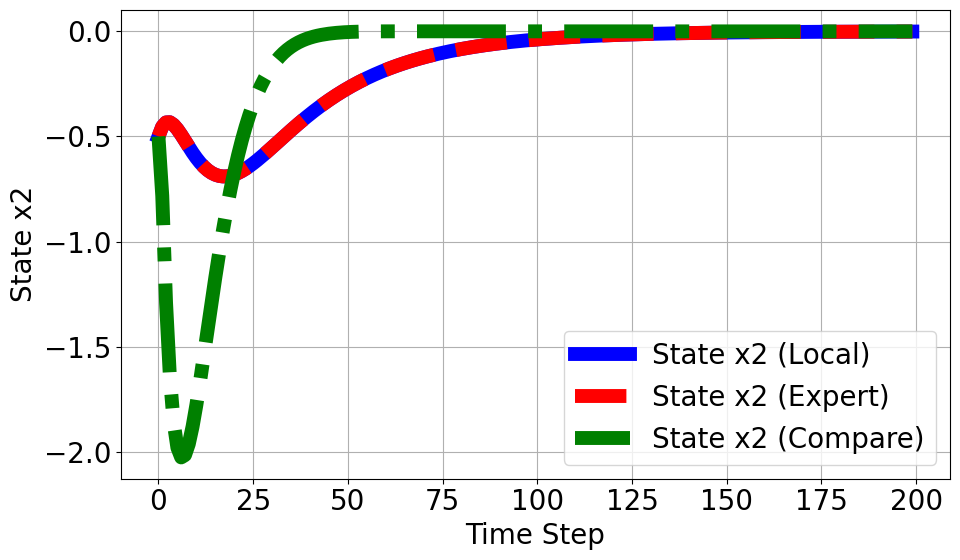}
        \caption{Generator mechanical power output $\Delta \mathcal{P}_m$}
         \label{fig:3b}
    \end{subfigure}%
    \hfill
    \begin{subfigure}[b]{0.33\textwidth}
        \centering
        \includegraphics[width=\linewidth]{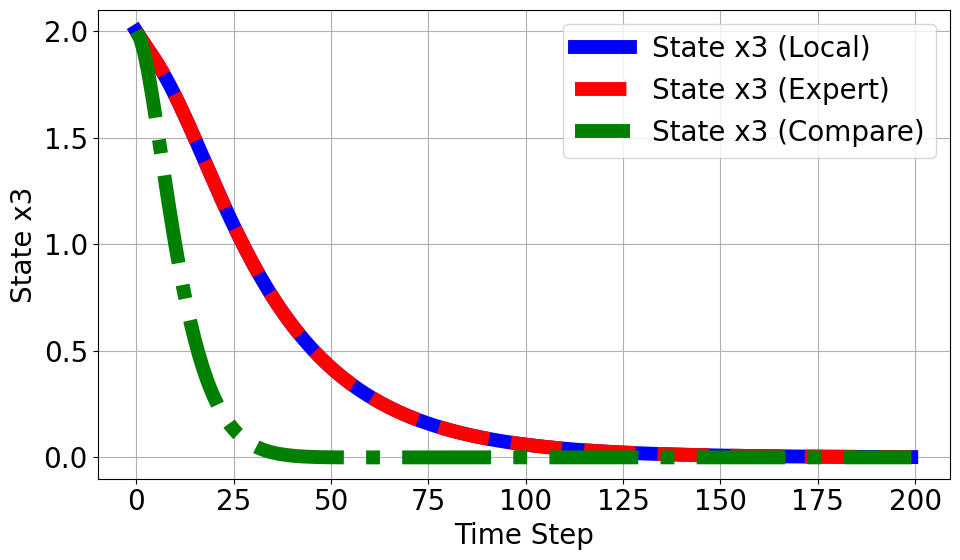}
        \caption{Frequency deviation $\Delta f_{G}$}
        \label{fig:3c}
    \end{subfigure}%
    \caption{Closed-loop state trajectories in the perturbed case.}
    \label{fig:3}
\end{figure*}

\subsection{Robust design}
In this subsection, Algorithm~\ref{Alg 4} is used to learn a robust generalized cost from randomly generated perturbation matrices $D_j$, where each entry is sampled from the Gaussian distribution $\mathcal{N}(0,\sigma=1.0)$. The hyperparameters are chosen as batch size $N_b=526$, maximum iteration number $T=7000$, and $R=1$. Figure~\ref{fig:4} shows the evolution of the Frobenius norms of $Q$ and $N$, together with the mini-batch mean loss. All three curves approach steady values after around 6000 iterations, indicating convergence of both the learned cost parameters and the training objective.

\begin{figure*}
    \centering
    \begin{subfigure}[b]{0.33\textwidth}
        \centering
        \includegraphics[width=\linewidth]{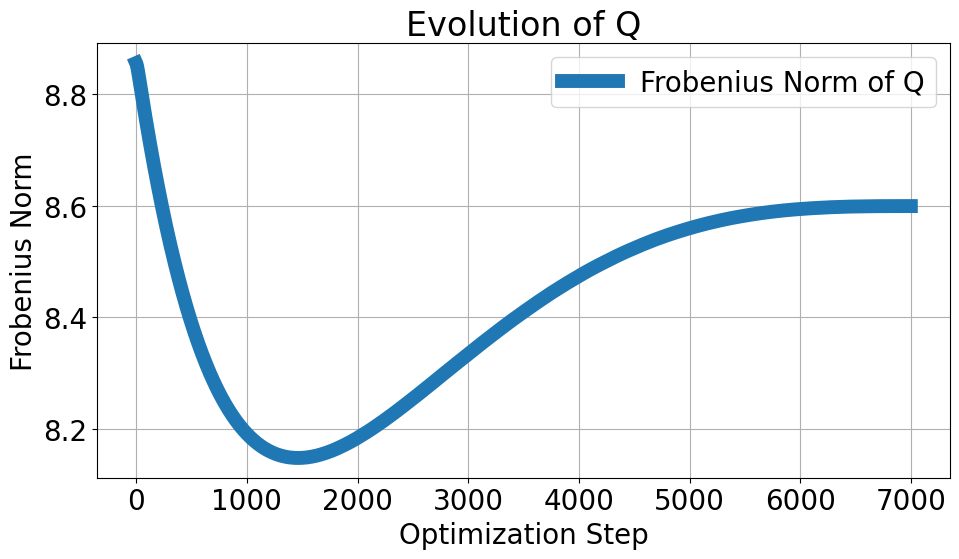}
        \caption{$\|Q\|_F$}
        \label{fig:4a}
    \end{subfigure}%
    \hfill
    \begin{subfigure}[b]{0.33\textwidth}
        \centering
        \includegraphics[width=\linewidth]{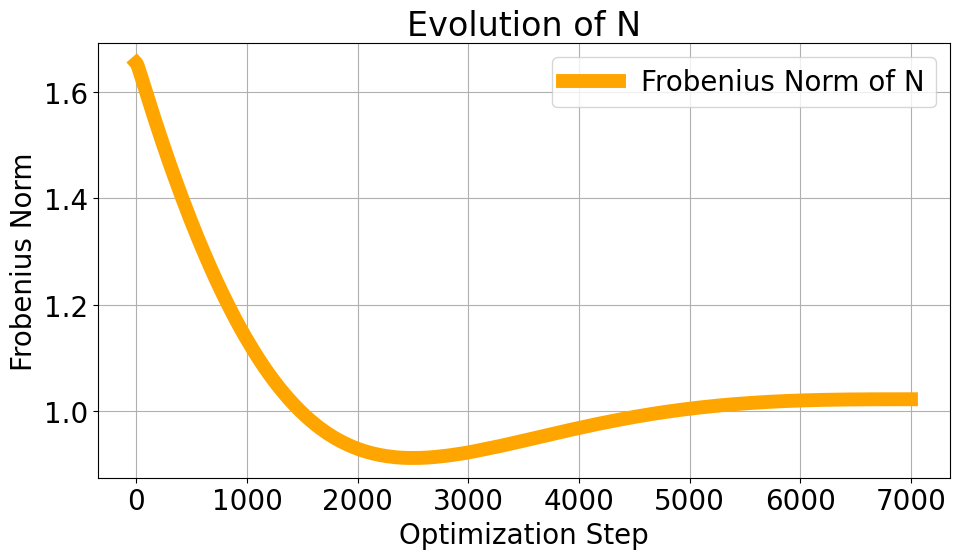}
        \caption{$\|N\|_F$}
         \label{fig:4b}
    \end{subfigure}%
    \hfill
    \begin{subfigure}[b]{0.33\textwidth}
        \centering
        \includegraphics[width=\linewidth]{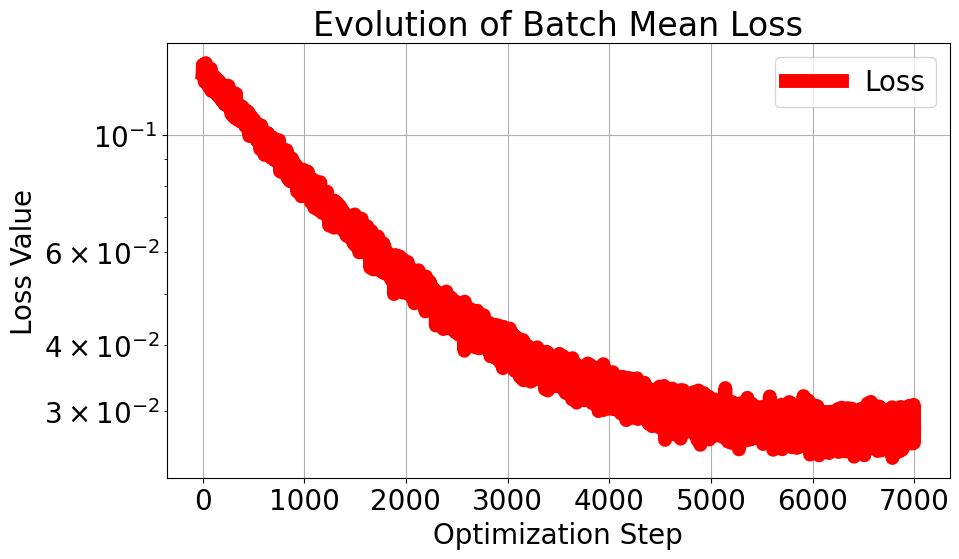}
        \caption{Mini-batch loss}
        \label{fig:4c}
    \end{subfigure}%
    \caption{Training history of the robust inverse-RL design.}
    \label{fig:4}
\end{figure*}

The final robust cost pair obtained after training is
\[
Q=\begin{bmatrix}
 3.108 &  1.013 & -0.385\\
 1.013  & 1.155 & 1.534\\ 
 -0.385 & 1.534  & 7.476  
\end{bmatrix}
\]
\[
N= [0.424 \quad  -0.111 \quad 0.923].
\]

To evaluate robustness, we compare the controller induced by the learned cost with the controller obtained from the true expert cost under three perturbation levels, namely $\sigma=1.0$, $\sigma=2.5$, and $\sigma=5.5$. Figure~\ref{fig:5} reports the corresponding state trajectories. At the training noise level $\sigma=1.0$, the proposed robust controller exhibits a slightly larger bias in the mean response than the controller derived from the true expert cost, but its trajectory variance is consistently smaller for all three states. This behavior is expected because Algorithm~\ref{Alg 4} is designed to reduce sensitivity to perturbations rather than to fit a single nominal system exactly.

As the perturbation level increases to $\sigma=2.5$ and $\sigma=5.5$, the advantage in variance reduction is preserved, while the mean bias of the robust controller becomes comparable to that of the true expert-cost controller. Hence, in the high-noise regime, the proposed method achieves nearly the same bias as the expert-cost design but with noticeably smaller dispersion.
 
\begin{figure*}
    \centering
    \begin{subfigure}[b]{0.33\textwidth}
        \centering
        \includegraphics[width=\linewidth]{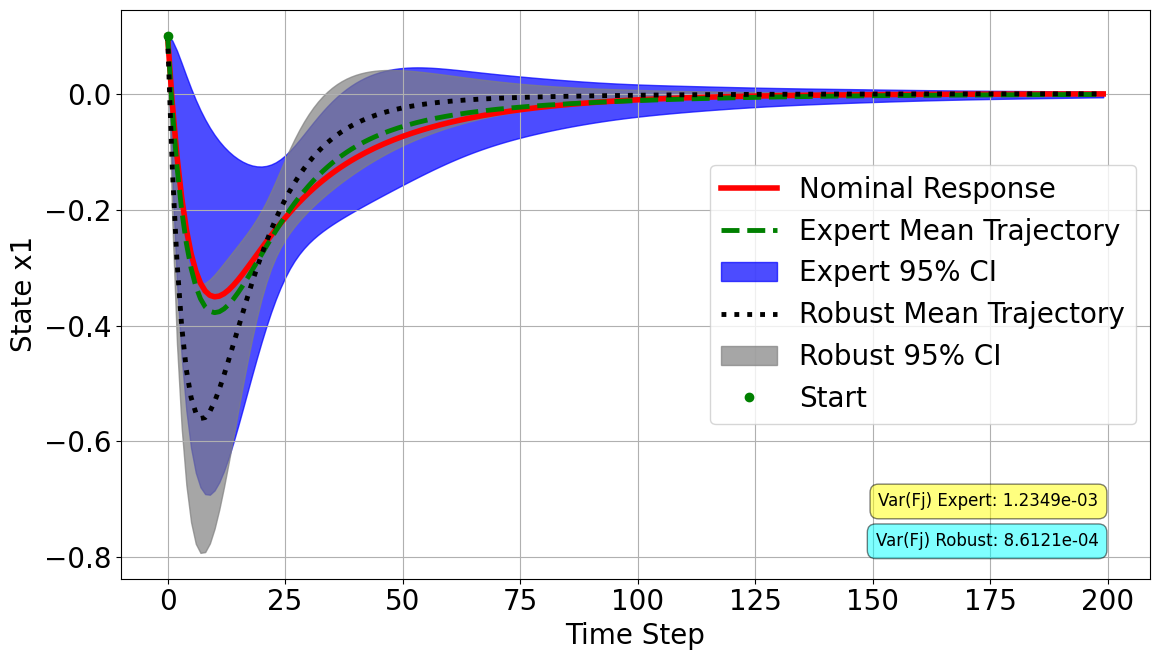}
        \caption{$\Delta \zeta$, $\sigma=1.0$}
        \label{fig:5a}
    \end{subfigure}%
    \hfill
    \begin{subfigure}[b]{0.33\textwidth}
        \centering
        \includegraphics[width=\linewidth]{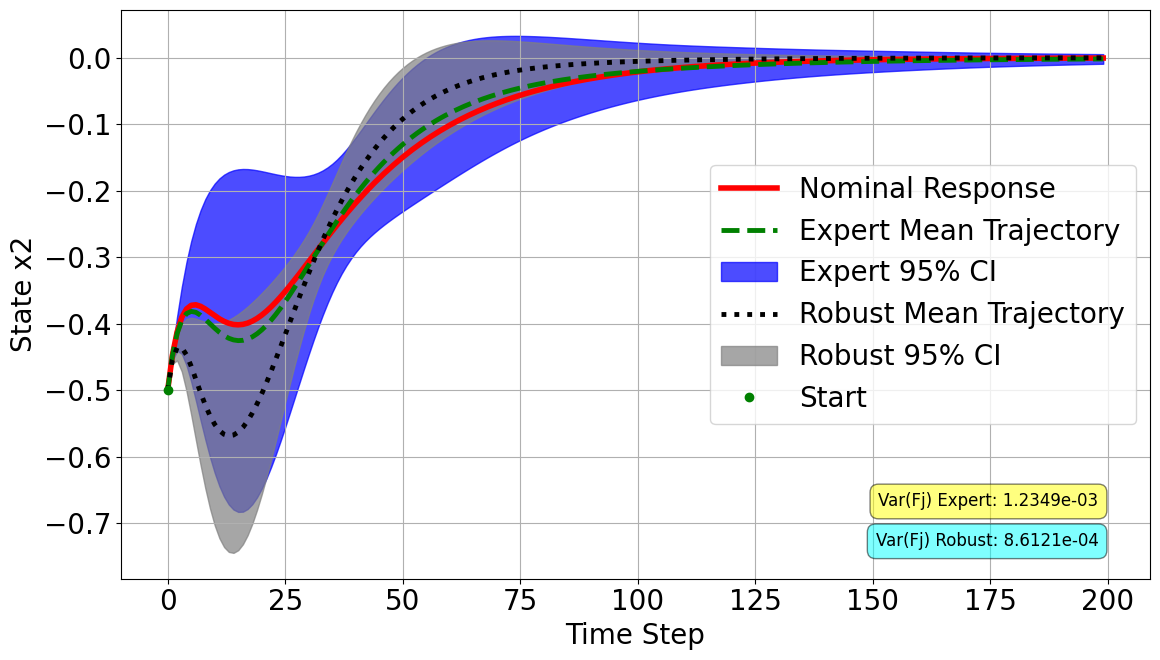}
        \caption{$\Delta \mathcal{P}_m$, $\sigma=1.0$}
         \label{fig:5b}
    \end{subfigure}%
    \hfill
    \begin{subfigure}[b]{0.33\textwidth}
        \centering
        \includegraphics[width=\linewidth]{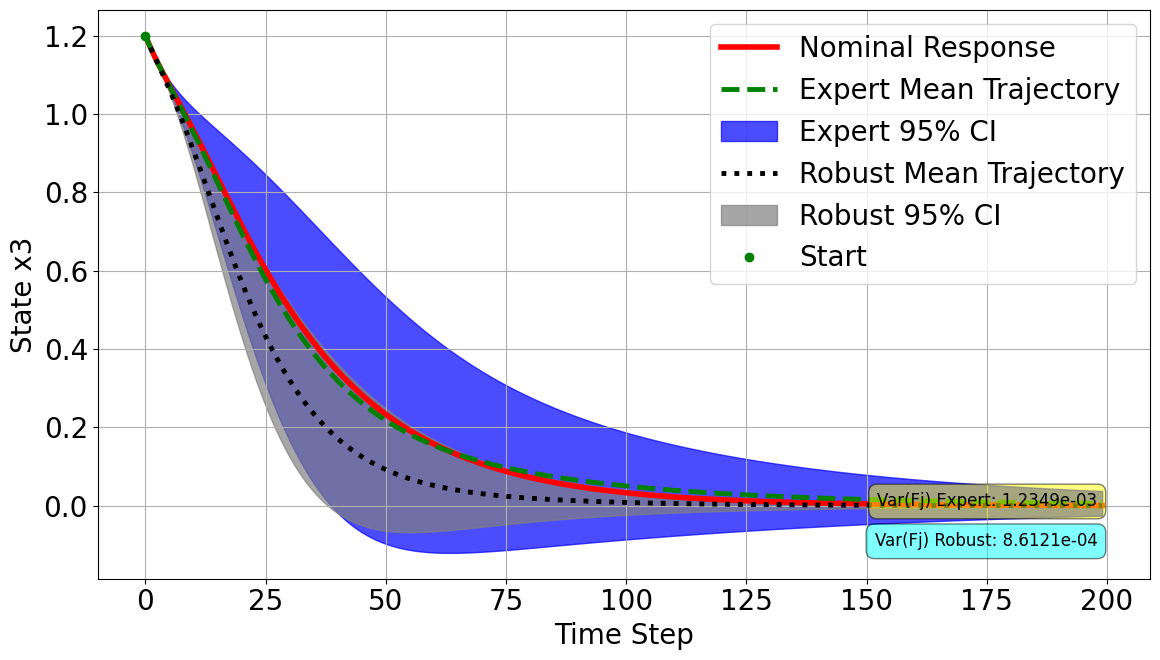}
        \caption{$\Delta f_G$, $\sigma=1.0$}
        \label{fig:5c}
    \end{subfigure}%
     \hfil
    \begin{subfigure}[b]{0.33\textwidth}
        \centering
        \includegraphics[width=\linewidth]{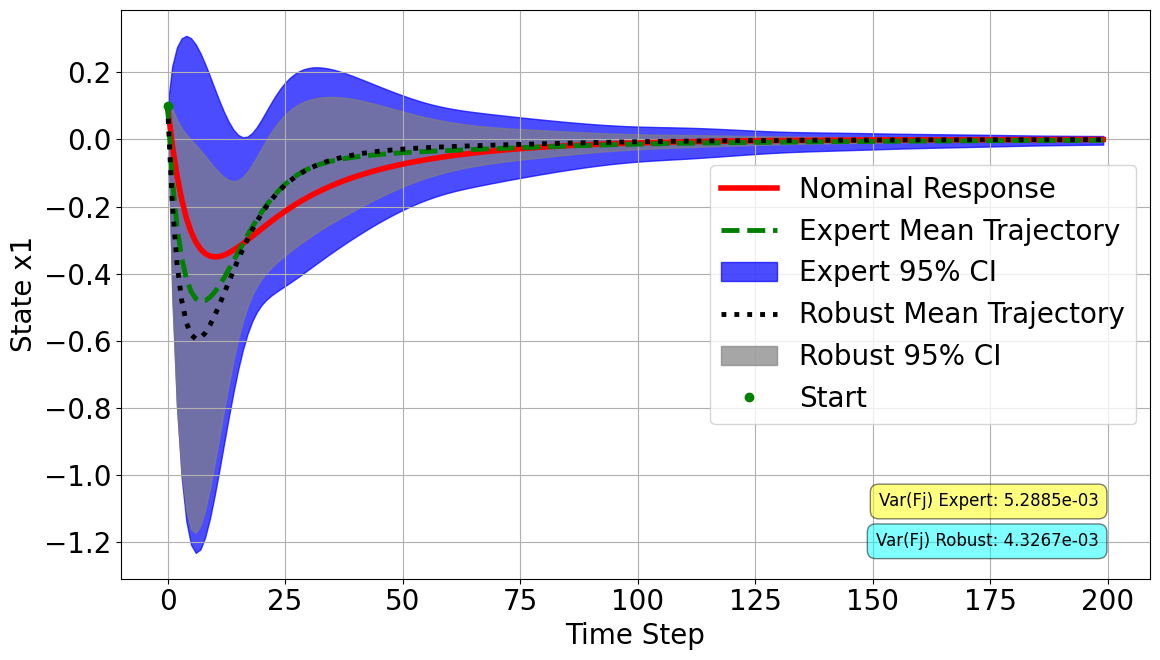}
        \caption{$\Delta \zeta$, $\sigma=2.5$}
        \label{fig:52a}
    \end{subfigure}%
    \hfill
    \begin{subfigure}[b]{0.33\textwidth}
        \centering
        \includegraphics[width=\linewidth]{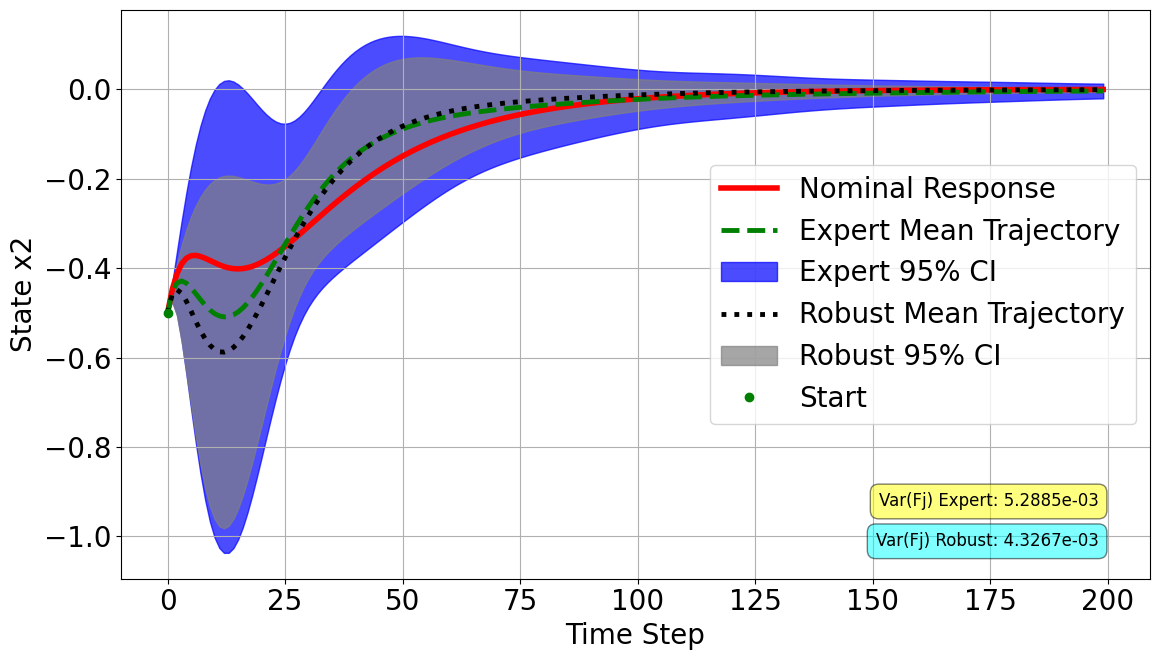}
        \caption{$\Delta \mathcal{P}_m$, $\sigma=2.5$}
         \label{fig:52b}
    \end{subfigure}%
    \hfill
    \begin{subfigure}[b]{0.33\textwidth}
        \centering
        \includegraphics[width=\linewidth]{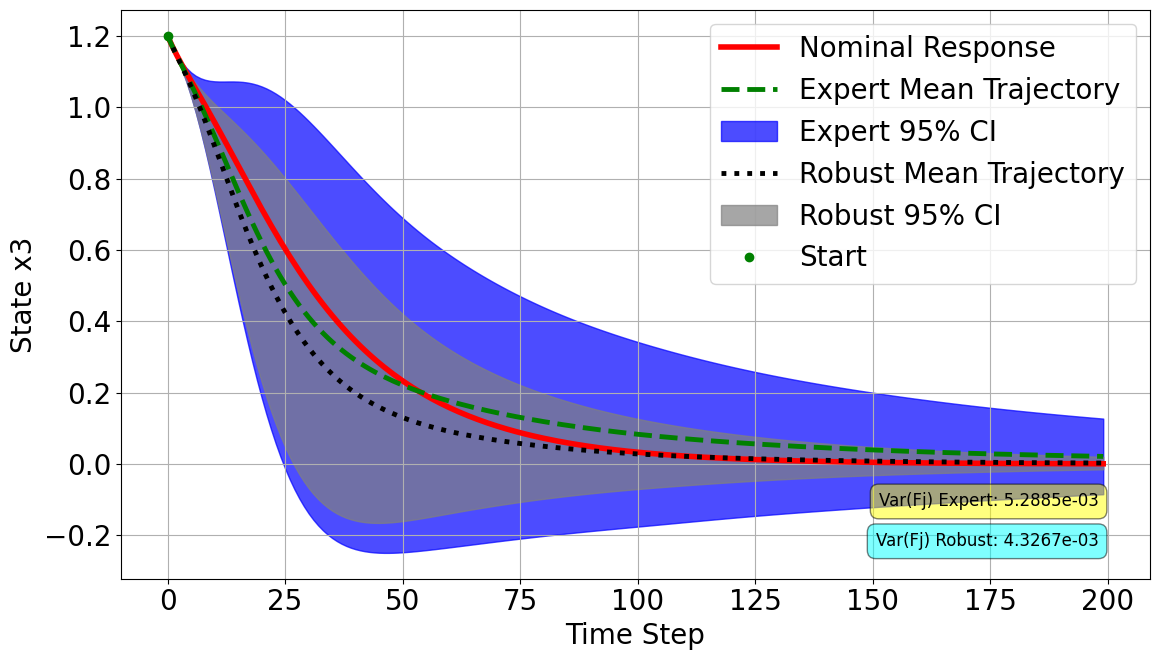}
        \caption{$\Delta f_G$, $\sigma=2.5$}
        \label{fig:52c}
    \end{subfigure}%
     \hfil
    \begin{subfigure}[b]{0.33\textwidth}
        \centering
        \includegraphics[width=\linewidth]{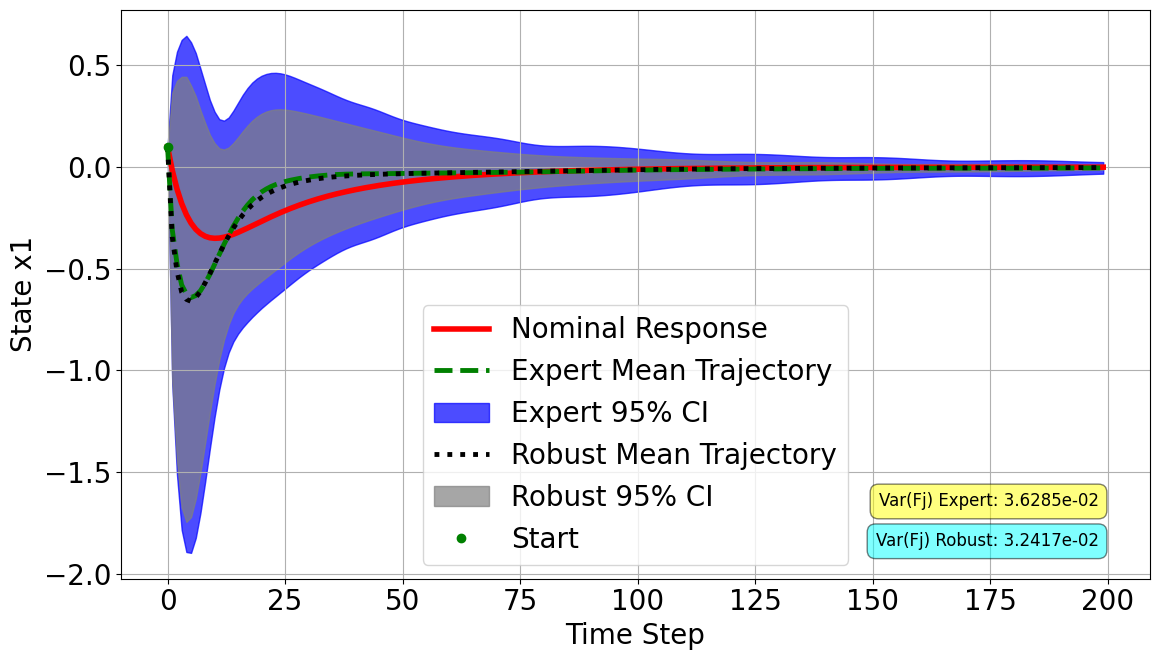}
        \caption{$\Delta \zeta$, $\sigma=5.5$}
        \label{fig:53a}
    \end{subfigure}%
    \hfill
    \begin{subfigure}[b]{0.33\textwidth}
        \centering
        \includegraphics[width=\linewidth]{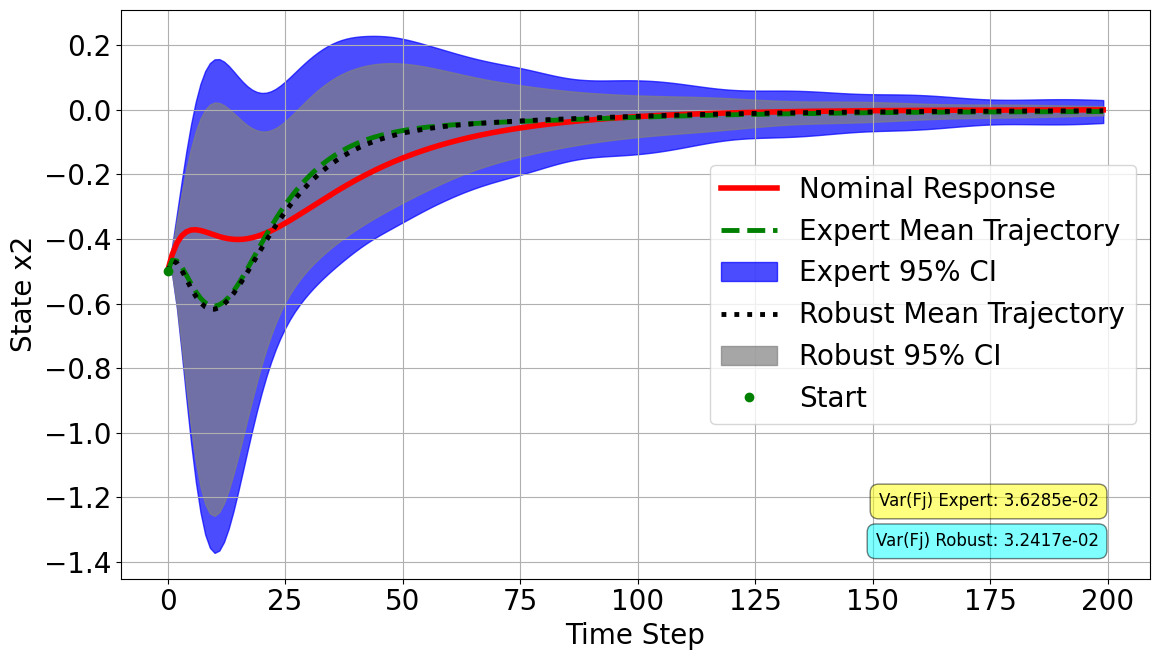}
        \caption{$\Delta \mathcal{P}_m$, $\sigma=5.5$}
         \label{fig:53b}
    \end{subfigure}%
    \hfill
    \begin{subfigure}[b]{0.33\textwidth}
        \centering
        \includegraphics[width=\linewidth]{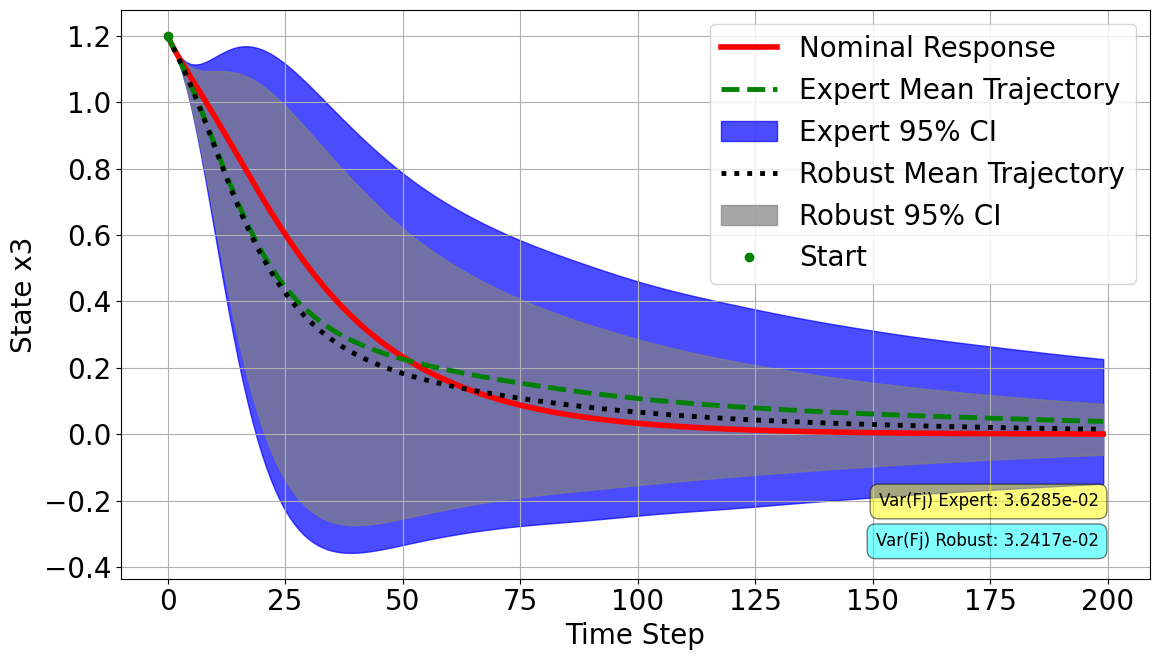}
        \caption{$\Delta f_G$, $\sigma=5.5$}
        \label{fig:53c}
    \end{subfigure}%

    \caption{State trajectories of the robust inverse-RL controller under different perturbation levels.}
    \label{fig:5}
\end{figure*}

\section{Conclusion}

This paper proposed a convex-optimization framework for data-driven inverse reinforcement learning of discrete-time linear systems with model uncertainty. For nominal systems, we derived a model-based semidefinite characterization of inverse optimality and a relaxed convex formulation that recovers an equivalent cost matrix together with a stabilizing controller from expert trajectory data. We then developed a model-free, off-policy reformulation by replacing the unknown system matrices with a regressed kernel matrix identified from local input--state data. For uncertain local systems, we showed that a standard LQR cost is generally too restrictive and introduced a generalized LQR formulation with a state--input cross term, which led to a convex data-driven inverse-RL method capable of recovering both a generalized cost pair and a stabilizing controller.

The paper also formulated a robust inverse-RL design problem over a distribution of plant perturbations and solved it using differentiable semidefinite programming and stochastic approximation. Simulation results on a discrete-time power-system example demonstrated accurate recovery of expert closed-loop behavior in both nominal and perturbed settings. In particular, the robust design produced smaller trajectory variance than the controller induced by the true expert cost, while achieving comparable bias in the high-noise regime. Overall, the proposed methods provide a simpler and more numerically stable alternative to classical iterative inverse-RL schemes because they avoid repeated inverse-RL policy/value updates and do not require an initial stabilizing gain. Future work will focus on extensions to output-feedback settings, more general classes of uncertain and nonlinear systems, and stronger theoretical guarantees under noisy data.

\printcredits
\section{Declaration of Competing Interest}
The authors declare that they have no known competing financial interests or personal relationships that could have appeared to influence the work reported in this article.

\section{Data availability}
The data that has been used is confidential.
\section{Acknowledgement}
With sincere thanks to the editors and reviewers for their valuable comments and suggestions, which contributed to the improvement of this paper.

\bibliography{cas-refs}
\bibliographystyle{elsarticle-num-names}


\end{document}